\documentclass[letterpaper,twocolumn,10pt]{article}
\usepackage{usenix-2020-09}

\usepackage{amssymb}
\usepackage{amsmath,amsfonts}
\usepackage{tikz}
\usepackage{booktabs}
\usepackage{amsmath}
\usepackage[english]{babel}
\usepackage{blindtext}
\usepackage{multirow}
\usepackage{cleveref}
\usepackage{enumitem}
\usepackage{amsmath}
\usepackage{amsthm}
\usepackage[ruled,vlined]{algorithm2e} 
\usepackage{graphicx} 
\usepackage{subfigure} 
\usepackage{algorithmic}
\usepackage{filecontents}
\newtheorem{theorem}{Theorem}
\newtheorem{definition}{Definition}

\begin{document}
\pagestyle{empty}  
\date{}

\title{\Large \bf A Fast Solver-Free Algorithm for Traffic Engineering\\ in Large-Scale Data Center Network}

\author{%
\textnormal{Yingming Mao\textsuperscript{1,2}, 
Qiaozhu Zhai\textsuperscript{1}, 
Ximeng Liu\textsuperscript{3}, 
Zhen Yao\textsuperscript{4}, 
Xia Zhu\textsuperscript{4}, 
Yuzhou Zhou\textsuperscript{1}}\\\textnormal{
\textsuperscript{1}Xi'an Jiaotong University}
\textnormal{\textsuperscript{2}Shanghai Innovation Institute}
\textnormal{\textsuperscript{3}Shanghai Jiao Tong University}
\textnormal{\textsuperscript{4}\textnormal{Huawei}}
}

\maketitle

\begin{abstract}
Rapid growth of data center networks (DCNs) poses significant challenges for large-scale traffic engineering (TE). Existing acceleration strategies, which rely on commercial solvers or deep learning, face scalability issues and struggle with degrading performance or long computational time.

Unlike existing algorithms adopting parallel strategies, we propose Sequential Source-Destination Optimization (SSDO), a sequential solver-free algorithm for intra-DCN TE. SSDO decomposes the problem into subproblems, each focused on adjusting the split ratios for a specific source-destination (SD) demand while keeping others fixed. To enhance the efficiency of subproblem optimization, we design a Balanced Binary Search Method (BBSM), which identifies the most balanced split ratios among multiple solutions that minimize Maximum Link Utilization (MLU). SSDO dynamically updates the sequence of SDs based on real-time utilization, which accelerates convergence and enhances solution quality.

We evaluate SSDO primarily on Meta DCNs, and additionally on two WAN topologies as auxiliary demonstrations of generality. In a Meta topology, SSDO achieves a 65\% and 60\% reduction in normalized MLU compared to TEAL and POP, two state-of-the-art TE acceleration methods, while delivering a $12\times$ speedup over POP. These results demonstrate the superior performance of SSDO in large-scale TE.
\end{abstract}

\section{Introduction}
With the rapid development of social networks  and large language models (LLMs) \cite{qian_alibaba_2024}, data center networks (DCNs) face increasingly demanding performance requirements. To address this, companies like Microsoft \cite{hong2013achieving} and Google \cite{poutievski_jupiter_2022} have adopted centralized Traffic Engineering (TE) systems powered by Software-Defined Networking (SDN) \cite{6567024,AKYILDIZ20141,10.1145/3341302.3342069,jain2013b4,liu_traffic_2014,10.1145/3482898.3483354,BwE}. These systems optimize traffic routing across fixed network paths to improve performance, often formulating TE as multicommodity flow problems \cite{mitra_case_1999} to minimize Maximum Link Utilization (MLU) or maximize network flow, solved by a centralized controller \cite{alizadeh_conga_2014, guo_dcell_2008, al-faresHederaDynamicFlow2010}.

The TE controller operates by collecting traffic demands and solving a linear programming \cite{mitra_case_1999} to determine traffic allocations. This periodic process ensures that the routing of traffic aligns with real-time demands  \cite{zhang_gemini_2021, poutievski_jupiter_2022}.  However, as DCNs scale to hundreds of nodes and tens of thousands of edges, the computational overhead grows significantly, making real-time TE increasingly challenging.

Contemporary traffic engineering (TE) acceleration methods can be broadly categorized into two approaches: linear programming (LP)-based and deep learning (DL)-based methods. LP-based algorithms accelerate TE by decomposing the TE optimization problem into smaller subproblems based on demands or topologies \cite{narayanan_solving_2021, abuzaid_contracting_2021}, which are solved concurrently. However, this often results in degrading TE quality, as neglecting the coupling between subproblems. DL-based methods, such as Teal \cite{xu_teal_2023} leverage historical data to directly map traffic matrices to TE configurations, significantly accelerating the computation process. However, these methods face challenges such as dependence on the quality and diversity of training data and may struggle to generalize to unseen traffic patterns or network conditions.

In contrast to conventional acceleration algorithms, our key insight is to address the coupling between subproblems by solving them in a carefully designed sequence, where each subproblem builds on the solution of the previous one. Unlike parallel schemes, which often struggle to maintain global coherence, the sequential strategy progressively incorporates global network information by following a structured optimization order. This iterative refinement stabilizes at a high-quality solution while mitigating the degradation issues that commonly hinder parallel methods, making it a more reliable alternative.

Sequential TE algorithms require each subproblem to be solved efficiently, as cumulative computation time can become a bottleneck. Thus, Sequential Source-Destination Optimization (SSDO) was proposed, which decomposes the original problem into subproblems, each optimizing the split ratios for a specific source-destination (SD). This structure enables the design of a binary search-based algorithm, avoiding the high complexity of LP solvers. However, subproblems often have multiple valid solutions, and selecting an unsuitable one can slow convergence and degrade TE quality, making LP solvers unsuitable for subproblem solving. To mitigate this, we develop the Balanced Binary Search Method (BBSM), which not only accelerates subproblem solving but also ensures that selected solutions enhance subsequent optimization.

In addition, SSDO adopts a dynamic optimization sequence that prioritizes edges with the highest utilization. In each iteration, it identifies the most congested edges and selects all SDs whose paths traverse them. The corresponding subproblems are then solved to adjust split ratios, reducing congestion. After each step, edge utilization is updated to guide subsequent optimizations, ensuring that SSDO continuously focuses on the most constrained parts of the network and accelerates convergence toward higher-quality solutions. Moreover, since SSDO ensures a non-increasing MLU during optimization, terminating the algorithm at any point guarantees a solution that is at least improved compared to the initial configuration.

We evaluate SSDO with Meta DCNs and various wide-area network (WAN) topologies. SSDO offers a better balance of computation time and TE quality than existing algorithms. On the Meta-Web topology with a per-pair four-path limit, SSDO reduces solution time by 92\% relative to LP, with an error of less than 1\%. It also reduces error by 60\% and time by 90\% against POP \cite{narayanan_solving_2021}, a state-of-the-art LP-based acceleration method. In topologies that are too large for DL-based methods, SSDO consistently delivers efficient and high-quality solutions. Our code is available at~\cite{ymmao-xjtusiiYingmingMaoSSDO2025}.

\textbf{This work does not raise any ethical issues.
}
\section{Background and Motivation}\label{motivation}

\subsection{Existing Methods Facing Scale Challenge}

The rapid expansion of networks has made large-scale TE increasingly challenging. As an LP problem, allocating traffic across paths containing hundreds of nodes often requires several hours using commercial solvers. Consequently, operators are seeking methods to accelerate TE optimization.

\noindent\textbf{LP-based direct methods.} 
Traditionally, TE is modeled as a multicommodity flow problem \cite{mitra_case_1999} and solved using commercial LP solvers due to its modest scale in earlier networks. However, with the expansion of data center networks, the computational overhead of LP solvers has become prohibitive. The worst-case complexity of LP is regarded as $O(n^{3})$ ( $O(n^{2.373})$ in \cite{narayanan_solving_2021,lee_efficient_2015}), making it impractical for large-scale networks. For example, in a fully connected network with 150 nodes, assuming four paths per SD, LP requires solving for $4 \times 150 \times 149 = 89,400$ variables. This leads to substantial memory usage and long computational times.	Commercial solvers attempt to accelerate computations by launching multiple threads, each running a different optimization algorithm independently. The solver then selects the solution from the fastest-converging algorithm. However, their acceleration relies on executing multiple optimization methods in parallel and selecting only the fastest one, which inherently limits performance improvements.

\noindent\textbf{DL-based direct methods.} 
DL approaches, such as DOTE \cite{perry_dote_nodate} and Figret \cite{liu_figret_2023}, have been introduced to accelerate TE using MLU as the loss function. Although these methods demonstrate efficiency in limited-scale DCNs, their performance deteriorates significantly at larger scales. For example, in the same scenario of 89,400 variables, the DL model must output all variables in the output layer, which greatly hampers its generalization due to the "curse of dimensionality" \cite{koppen_curse_2000}. This constraint makes DL-based direct methods ill-suited for scaling up to large network sizes.

\noindent\textbf{LP-Based parallel accelerating methods.}
Parallel methods have emerged as promising solutions to accelerate TE
processes. For example, the POP method \cite{narayanan_solving_2021}
decomposes the optimization problem into $k$ subproblems, each preserving the network
topology but handling only a subset of demands. Similarly, NCFlow \cite{abuzaid_contracting_2021}
partitions both the demands and the network topology into $k$ distinct clusters.
These methods solve all subproblems simultaneously by invoking LP solvers and then combine their solutions to approximate an acceptable feasible solution. Increasing
$k$ can significantly reduce computational time, but this comes at the cost
of degrading TE performance due to the coupling between subproblems.
This trade-off between computation time and solution quality is a critical limitation
of parallel LP-based approaches.

\noindent\textbf{DL-based parallel accelerating methods.} 
To alleviate the “curse of dimensionality” in DL methods, Teal \cite{xu_teal_2023} was introduced. Similar to POP, Teal utilizes a shared policy network to independently compute split ratios for each demand. Additionally, Teal incorporates a multi-agent reinforcement learning (MARL) strategy to manage coupling among demands. Despite its advancements, the efficacy of Teal is significantly dependent on the correlation between historical and future traffic matrices and the generalizability of the shared policy network. These factors may result in degradation within complex network environments.

\subsection{Accelerate TE with Sequential Strategy}

Due to the difficulty of parallel strategies in addressing the coupling between subproblems, we propose a sequential strategy to optimize traffic allocation. By decomposing the problem into subproblems, each modifying the split ratios for a specific SD, and determining an appropriate solving order, the sequential strategy has the potential to achieve higher-quality traffic allocations compared to parallel strategies.

\noindent\textbf{Better handling of subproblem coupling.}
Unlike parallel methods that solve subproblems simultaneously but struggle with global coherence, our approach addresses subproblems sequentially, with each decision based on the previous one. This allows each subproblem to progressively capture the overall state of the network.  By structuring the solving order, the sequential approach better accounts for subproblem coupling, leading to better performance than parallel methods.

\noindent\textbf{Direct inheritance of existing algorithm Results.} 
Due to the monotonic nature of the proposed sequential algorithm, when initialized with a TE configuration derived from other methods, the resulting performance will always be at least as good as the original configuration. This ensures compatibility with previous approaches while enabling further improvement.

\noindent\textbf{Leveraging all available computing time.} The adjustment cycles for split ratios vary significantly across different networks, ranging from 10 seconds to 15 minutes, posing challenges for TE. LP-based parallel approaches require selecting $k$, the number of subproblems, to fit within the given cycle. However, a smaller $k$ improves precision but increases complexity, potentially exceeding the adjustment cycle, while a larger $k$ simplifies subproblems but sacrifices precision, degrading solution quality. Similarly, DL-based methods, while fast, inherently lack mechanisms to utilize unused computing time for further refinement. Once the solution is computed, any remaining adjustment time is left idle. In contrast, SSDO adapts seamlessly to varying adjustment cycles by performing high-frequency updates to split ratios starting from an initial feasible TE configuration. This approach ensures consistent improvement for short cycles while fully utilizing longer cycles for further refinement, enabling superior configurations under different computation time constraints.

\subsection{Key Challenges in Designing Effective Sequential Strategies}
While sequential strategies have the potential to achieve high-quality solutions, their implementation presents significant challenges. Designing an effective sequential approach requires addressing key issues related to computation time, solution consistency, and task sequencing.

\noindent\textbf{Computing efficiency for subproblems.} 
Sequential strategies solve subproblems one by one, making efficiency critical, especially when the number of subproblems is large. Although commercial solvers such as CPLEX  and Gurobi  offer efficient methods for solving optimization problems, their overhead in model construction and complex solving processes make them impractical for handling individual subproblems in a sequential framework.

\noindent\textbf{Inconsistency between subproblem and global performance.} 
Decisions made in early subproblems can constrain the solution space for later ones, potentially leading to poor global performance. This lack of coordination often results in inferior overall outcomes, requiring additional adjustments to improve global performance.

\noindent\textbf{Impact of subproblem order.} 
The sequence in which subproblems are solved significantly affects convergence speed and solution quality. While a random order can yield improvements over initial conditions, an inefficient order may slow convergence, requiring more iterations to achieve satisfactory results. Identifying an effective order is critical for improving solution quality and computational efficiency.

\section{TE MODEL}
\setlist[itemize]{leftmargin=*}
\noindent\textbf{Notations \& Definitions}: We present the recurrent mathematical notations and definitions pertaining to TE. As this work focuses on TE in DCNs, we focus on one-hop and two-hop transit paths, which suffice for most DCN deployments \cite{zhang_gemini_2021,poutievski_jupiter_2022}. The generalization to multi-hop scenarios, typical of WANs, is given in Appendix \ref{Notation for multi}.
\begin{itemize}
\item \textbf{Network.} The network topology is a graph $G = (V, E, c)$, with $V$ as vertices, $E$ as edges and $c_{ij}$ specifying the sum of capacities from vertices $i$ to $j$.

\item \textbf{Traffic demands.} The Demand matrix, denoted as $D$, is a $|V| \times |V|$ matrix where each element $D_{ij}$ represents the traffic demand from source $i$ to destination $j$.

\item \textbf{TE configuration.} TE configuration $\mathcal{R}$ outlines the split ratio, indicated as $f_{ikj}$, which expresses the proportion of traffic from the source $i$ to the destination $j$ that crosses an intermediary node $k$. This 3D matrix compactly encodes split-ratio information, representing the distribution of traffic across routing paths. Formally:
\begin{itemize}
    \item $f_{ikj}$: Represents the fraction of traffic from $i$ to $j$ that follows a two-hop path through $k$, where $i \neq k \neq j$.
    \item $f_{ijj}$: Denotes the fraction of traffic directly routed from $i$ to $j$ (1-hop path), where $i \neq j$.
    \item $f_{iij}$ and $f_{iki}$: Since the direct path is already captured by $f_{ijj}$, and self-traffic is not considered, $f_{iij} = f_{iki} = 0$.
\end{itemize}
This 3D matrix stores split ratio information densely, providing a strong basis for future calculations.

\item \textbf{Path set.} Practical TE systems typically constrain the set of paths available between SDs due to network topology or operational policies. The path set, denoted as $\mathcal{P}$, represents all permissible routing paths. Each element of $\mathcal{P}$ is an ordered triad of nodes, such as $(s, k, d)$, representing a valid path between source $s$ and destination $d$ via intermediate node $k$. If traffic follows a direct path, we set \( k = d \). For a given $(s, d)$, we define $\mathcal{K}_{sd}$ as the set of intermediate nodes $k$ associated with the paths in $\mathcal{P}$. Specifically, $\mathcal{K}_{sd} = \{k \mid (s, k, d) \in \mathcal{P}\}$.

\item \textbf{TE objective.} The objective function explored in this study aims to minimize MLU, denoted $u$, a metric widely used in TE \cite{azar_optimal_2003,benson_microte_2011,chiesa_lying_2016,poutievski_jupiter_2022,valadarsky_learning_2017}.  It effectively encapsulates both throughput and resilience to traffic fluctuations. MLU is defined as $\max_{i,j\in V}{(\sum_{k\in V} f_{ijk} \cdot D_{ik} + \sum_{k\in V} f_{kij} \cdot D_{kj})/c_{ij}} $, which is calculated by the given Demand matrix $D$ and the TE configuration $\mathcal{R}$.

\end{itemize}
\noindent\textbf{Optimization model of TE}:
The TE problem can be formulated as a linear programming (LP) problem, where the goal is to determine the optimal split ratios to minimize MLU while satisfying flow conservation constraints. The optimization model is defined as Equation (\ref{eq:TE}).
\begin{equation}
\begin{aligned}
    & \min_{f_{ikj}\in \mathcal{R}} u \\
    & \text{s.t.} \quad 
    \begin{cases}
        f_{ikj} \geq 0, \quad f_{iki} = 0, \quad f_{iij} = 0, & \forall i, j, k\in V, \\
        f_{ikj}=0,  &\forall (i,k,j)\notin \mathcal{P},
        \\
        \sum_{k\in V} f_{ikj} = 1, & \forall i \neq j \in V, \\\frac{\sum_{k\in V} f_{ijk} \cdot D_{ik} + \sum_{k\in V} f_{kij} \cdot D_{kj}}{c_{ij}}
         \leq u   , & \forall i \neq j \in V.
    \end{cases}
\end{aligned}
\label{eq:TE}
\end{equation}

\section{SSDO Design }
\subsection{Overview}
\begin{figure}[b]
    \centering
    \includegraphics[width=1\linewidth]{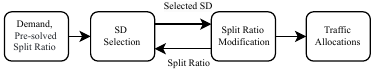}
    \caption{Workflow of SSDO.}
    \label{fig:Work flow of SSDO}

\end{figure}

As illustrated in Figure \ref{fig:Work flow of SSDO}, SSDO takes predetermined split ratios and traffic demand as input. The \emph{SD Selection} component identifies SDs based on the current split ratios and traffic demands. The \emph{Split Ratio Modification} component then optimizes the split ratios for the selected SD. This iterative process alternates between \emph{SD Selection} and \emph{Split Ratio Modification}. With each iteration, the system’s MLU progressively decreases, ultimately converging to a high-quality solution.

For a given SD $(s,d)$, the \emph{Split Ratio Modification} component formulates a subproblem with  $f_{skd}, \forall k \in \mathcal{K}_{sd}$  as decision variables, while keeping other split ratios fixed.  Instead of solving it as an LP problem, SSDO reformulates it as a structured search problem, significantly reducing computational complexity. LP solvers rely on costly matrix operations and iterative constraint satisfaction, often requiring  $O(n^{3})$  complexity for large-scale problems. In contrast, SSDO employs a binary search algorithm, which converges in logarithmic time with only a few function evaluations, avoiding the high overhead of traditional optimization techniques. To ensure subproblem solutions align with global TE performance, SSDO selects the most balanced solution among the subproblem’s optima. A detailed description is provided in \S\ref{Split Ratio Modification}.

The \emph{SD Selection} component in SSDO identifies the set of edges with maximal utilization, determined by the split ratios and demands. It then locates the associated SDs and provides them to the \emph{Split Ratio Modification} component.  Without a well-designed selection procedure, the process could converge slowly or settle into inferior local optima. SSDO’s carefully crafted rules largely avert these pitfalls, as we detail in \S\ref{SD selection}.

In addition, SSDO can be initiated with any feasible presolved split ratios. A potential approach to constructing this solution is to route each SD’s demand entirely along one of its available paths. All components of SSDO are meticulously designed, necessitating only basic matrix operations of addition and multiplication. SSDO does not require historical data or significant computational resources, making it straightforward to program and implement.

\subsection{Split Ratio Modification Component}\label{Split Ratio Modification}

\textbf{Subproblem definition.} In this section, we focus on an LP subproblem of TE. In the subproblem, only the split ratios related to the selected SD are subjected to optimization, while all other split ratios remain constant, which is called subproblem optimization (SO). To elucidate the fundamental concept of SO, the process is illustrated in Figure \ref{fig:SO-illustration}. Within this network, there are three SDs: $(A,B)$, $(B,C)$, and $(A,C)$. The initial TE scheme routes all traffic along the shortest paths, resulting in an MLU of  $\max\{1, 0.5, 0.5\} = 1$ , which occurs at the edge  $A \rightarrow B$ . By altering the split ratios for $(A,B)$ and maintaining those for $(B,C)$ and $(A,C)$ unchanged, the MLU transitions to $\max\{0.75,0.75,0.5,0.25\}=0.75$. In particular, 0.75 represents the minimum MLU achievable in this system under the given traffic pattern.
\begin{figure}[t]
    \centering
    \includegraphics[width=1\linewidth]{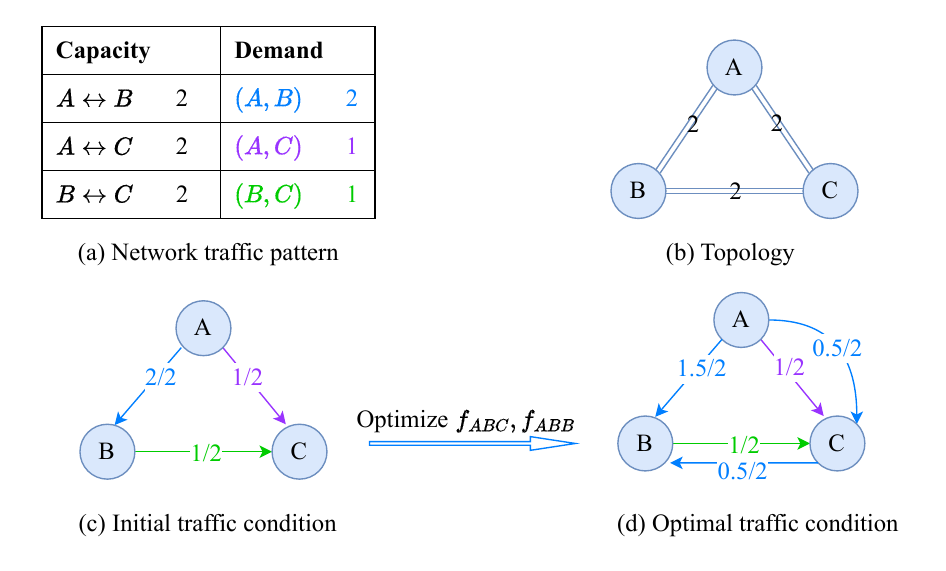}
    \caption{A sample illustration of the subproblem optimization (SO). Notations like “1/2” on the edges mean that the flow through the edge is 1 and the capacity of the edge is 2. In this example, only one SO is required for the SSDO algorithm. In initial TE scheme, $f_{ABB}=100\%$, $f_{ACB}=0\%$, $f_{ACC}=100\%$, $f_{ABC}=0\%$, $f_{BCC}=100\%$, $f_{BAC}=0\%$. After SO process, $f_{ABB}$ change to 75\%, $f_{ACB}$ change to 25\%.}
 
    \label{fig:SO-illustration}
\end{figure}

\noindent \textbf{Subproblem characteristics.} Compared to the original problem like Equation (\ref{eq:TE}), the SO problem of given SD $(s,d)$ requires optimizing only the $|\mathcal{K}_{sd}|$ split ratios, significantly simplifying the problem. From a programming perspective, the SO problem remains an LP problem. Fortunately, it has some unique characteristics that can be further leveraged to simplify the calculation.

\textbf{Characteristic 1: \textit{Without solving SO, the feasibility of a given MLU \( u_0 \) can be analytically judged.}}

The feasibility of a given MLU \( u_0 \) can be judged without solving SO. This process is illustrated in Figure \ref{fig:judgement} and involves the following steps:
\begin{enumerate}
    \item \textbf{Background traffic computation:} Suppose that the selected SD is designated as $(s, d)$ . Background traffic from node $i$ to node $j$, denoted as $ Q_{ij}$, can be determined by setting $f_{skd} = 0$ for all $k \in V$, as in Equation (\ref{eq:Qij1}). The calculation example is shown in (b) of Figure \ref{fig:judgement}.
\begin{equation}
    Q_{ij} = 
    \begin{cases} 
        \sum\limits_{k \in V} f_{ijk} \cdot D_{ik} + \sum\limits_{k \in V} f_{kij} \cdot D_{kj},&i \ne s, j \ne d\\
        \sum\limits_{k \in V/d} f_{ijk} \cdot D_{ik} + \sum\limits_{k \in V/s} f_{kij} \cdot D_{kj},&i = s, j = d \\
        \sum\limits_{k \in V/d} f_{ijk} \cdot D_{ik} + \sum\limits_{k \in V} f_{kij} \cdot D_{kj},&i = s, j \ne d \\
        \sum\limits_{k \in V} f_{ijk} \cdot D_{ik} + \sum\limits_{k\in V/s} f_{kij} \cdot D_{kj},&i \ne s, j = d
    \end{cases}
    \label{eq:Qij1}
\end{equation}

\item \textbf{Residual capacity calculation:} For a given path $s \rightarrow k \rightarrow d$, the residual capacity $T_{skd}$ is computed using Equation (\ref{eq:Tikj}). Here, $T_{skd}$ represents the maximum remaining capacity of the path, calculated by the background traffic $Q$ and the given $u_0$. Based on this residual capacity, the upper bound of the split ratio through $k$, denoted as ${\bar{f}}_{skd}$, is derived using Equation (\ref{eq:fikj}).
\begin{equation}
    T_{{s}k{d}} = 
    \begin{cases} 
        \min \left\{
        \begin{array}{l}
            {u_0}{c_{{s}k}} - {Q_{{s}k}}, \\
            {u_0}{c_{k{d}}} - {Q_{k{d}}}
        \end{array} 
        \right\}, & k\in \mathcal{K}_{sd},k \neq {d}, \\
        {u_0}{c_{{s}{d}}} - {Q_{{s}{d}}}, & k = {d}
    \end{cases}
    \label{eq:Tikj}
\end{equation}
\begin{equation}
    \bar{f}_{skd} = \frac{T_{skd}}{D_{sd}},
    \label{eq:fikj}
\end{equation}

\item \textbf{Feasibility assessment:} Drawing from the preceding analysis, the feasibility of SO can be evaluated through the following metrics. \begin{itemize}
        \item If \( {\sum_{k\in \mathcal{K}_{sd}} \bar{f}}_{skd} \geq 1 \) and \( {\min_{k\in \mathcal{K}_{sd}} \bar{f}}_{skd} \geq 0 \), there is a feasible solution. In this case, \( \bar{f}_{skd} \) can be normalized to determine the solution, as shown in Figure \ref{fig:judgement}.
        \item If \( {\sum_{k\in \mathcal{K}_{sd}} \bar{f}}_{skd} < 1 \) or \( {\min_{k\in \mathcal{K}_{sd}} \bar{f}}_{skd} < 0 \), the given \( u_0 \) lies outside the feasible domain.
    \end{itemize}

\end{enumerate}

\begin{figure}[ht]
    \centering
    \includegraphics[width=1\linewidth]{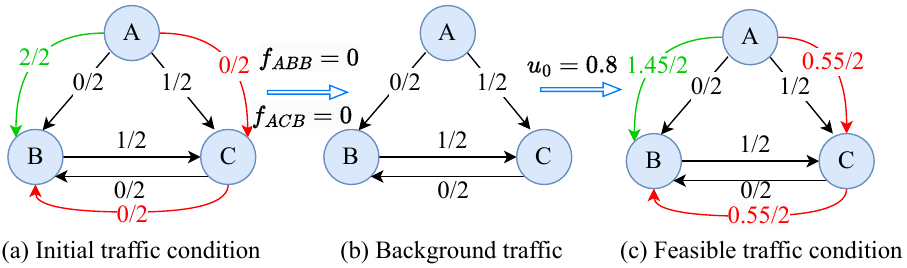}
    \caption{Illustration of judgment process of SO proposed in Figure \ref{fig:SO-illustration} when $u_0=0.8$, $D_{AB}=2$. The green and red lines represent the traffic flows on $A\rightarrow B$ and $A\rightarrow C\rightarrow B$ .  Let $f_{ABB}=f_{ACB}=0$, background traffic $Q$ is calculated in (b). Using background traffic, $T_{ACB}=\min{\{2\times0.8-1,2\times0.8-0\}}=0.6$, $T_{ABB}=0.8\times2=1.6$, ${\bar{f}}_{ACB}=0.6/2=0.3$, ${\bar{f}}_{ABB}=1.6/2=0.8$. To satisfy the normalization constraint, $f_{ACB},f_{ABB}=0.3/\left(0.8+0.3\right),0.8/\left(0.8+0.3\right)$, a feasible solution having been obtained in (c).}
   
    \label{fig:judgement}
\end{figure}

\textbf{Characteristic 2: \textit{The optimal MLU \( u^\ast \) in SO can be determined by binary search.}}

Based on the analysis above, the upper bound of the split ratio ${\bar{f}}_{skd}$ is fundamentally related to the MLU parameter $u$. As rigorously proven in Appendix \ref{appendix:monotonicity}, each individual $\bar{f}_{skd}(u)$ is a nondecreasing function of $u$. This component-wise monotonicity implies that for any intermediate node $k \in {\mathcal{K}_{sd}}$, we have:
\begin{equation}
    {\bar{f}}_{skd}(u) \geq {\bar{f}}_{skd}(u_0)\quad \text{whenever }  u\ge u_0.
    \label{eq:inequality}
\end{equation}
The aggregation of these monotonic components preserves the nondecreasing property. Specifically, summing over all possible paths $k \in {\mathcal{K}_{sd}}$ yields:
\begin{equation}
    \sum_{k \in \mathcal{K}_{sd}}{\bar{f}}_{skd}(u) \geq \sum_{k \in \mathcal{K}_{sd}}{\bar{f}}_{skd}(u_0)\quad \text{whenever }  u\ge u_0.
    \label{eq:sum_inequality}
\end{equation}
This monotonicity ensures that if \( u_0 \) is feasible, then all \( u \geq u_0 \) are also feasible. Conversely, if \( u_0 \) is infeasible, then all \( u \leq u_0 \) are also infeasible.

To perform a binary search, we must define the lower and upper boundaries \( u^{lb} \) and \( u^{ub} \), which ensure a bounded search space. These boundaries are given as Equation (\ref{eq:ulb}) and Equation (\ref{eq:uub}). \( u^{lb} \) represents the minimum possible MLU, below which the solution becomes infeasible. Specifically, for \( u < u^{lb} \), the split ratio \( {\bar{f}}_{skd} \) would become negative, violating the feasibility conditions. \( u^{ub} \) provides the maximum feasible MLU under initial conditions before modification. This ensures that any feasible solution must lie within \( [u^{lb}, u^{ub}] \).
\begin{equation}
    u^{lb} = \max_{i,j \in V} \frac{Q_{ij}}{c_{ij}},
    \label{eq:ulb}
\end{equation}
\begin{equation}
    u^{ub} = \max_{i,j \in V} \frac{\sum_{k\in V} f_{kij} \cdot D_{kj} + \sum_{k\in V} f_{ijk} \cdot D_{ik}}{c_{ij}}.
    \label{eq:uub}
\end{equation}

With these boundaries established, we can conclude that there exists a threshold \( u^\ast \in [u^{lb}, u^{ub}] \) such that \( u^\ast \) is the optimal MLU. The monotonicity of \( {\bar{f}}_{ikj}(u) \) further guarantees the correctness of the binary search within this range. Thus, the above analysis ensures that the binary search can determine not only feasible but also optimal MLU \( u^\ast \) in SO.

\textbf{Characteristic 3: \textit{For the optimal MLU \( u^\ast \), there exist multiple feasible TE configurations, but only one balanced TE configuration which can be binary searched.}}

As illustrated in Figure \ref{fig:multi-solution}, the multi-solution phenomenon for split ratios occurs if and only if the optimal MLU $u^\ast$ equals the lower bound u$^{lb}$. Otherwise ($u^\ast > u^{lb}$), the solution is unique. Under this specific condition, multiple sets of split ratios can achieve the same  $u^\ast$, resulting in ambiguity in the solution. To address this issue and better coordinate the performance of SO and origin optimization, we introduce ‘balance’ as a secondary objective in SO. The balanced solution is formulated to satisfy the following two key conditions.
\begin{itemize}
    \item For each path with non-zero split ratios, the maximum utilization of its edges equals a fixed value $u^e $.
    \item For each path with zero split ratios, the maximum utilization of its edges exceeds or equals  $u^e$ .
\end{itemize}

An example of this balanced solution is shown in (c) of Figure \ref{fig:multi-solution}.For (A,B), \( f_{ACB} \) and \( f_{ADB} \) are greater than zero, the maximum utilization of the paths \( A \rightarrow C \rightarrow B \) and \( A \rightarrow D \rightarrow B \) is equal to 0.55, satisfying the first condition. Furthermore, the maximum utilization of \( A \rightarrow B \) exceeds 0.55, fulfilling the second condition. In contrast, an alternative solution shown in (d) of Figure \ref{fig:multi-solution} fails to meet the second condition, as the maximum utilization of the paths $A\rightarrow D \rightarrow B$ does not exceed the threshold \( u^e \), highlighting its imbalance in this scenario. By ensuring that the balanced solution meets these conditions, it not only resolves the ambiguity caused by the multisolution phenomenon, but also guarantees a more balanced distribution of traffic across paths.

The introduction of \( u^e \) provides significant benefits in the optimization process.
\begin{itemize}
    \item \textbf{Providing more optimization potential.} Without \( u^e \), the SO process cannot effectively determine which solution among multiple feasible configurations is optimal for the overall TE objective. Blindly increasing the traffic on certain edges may severely restrict the optimization space for subsequent SDs, leading the algorithm to converge on inferior solutions. By balancing capacity utilization across edges, \( u^e \) helps avoid such pitfalls. Although it may result in a time cost for finding \( u^e \), this balanced approach ensures that the solution space remains flexible, preventing the algorithm from being trapped in poor feasible configurations.

\item \textbf{Seamless integration into the binary search framework.} 
Like the optimal MLU $u^\ast$, the balanced MLU $u^e$ can also be determined via binary search. 
The key difference is that $u^e$ is based on the modified upper bounds:
\begin{equation}
    \bar f_{skd}^b(u) = \max\{0, \bar f_{skd}(u)\}, 
    \label{eq:fbarb}
\end{equation}
 which exclude negative split ratios. 
The search space of $u^e$ is bounded between 0 and $u^{ub}$ (the same upper bound as $u^\ast$). 
In single-solution cases, $u^e$ coincides with $u^\ast$; in multi-solution cases, 
binary search on $u^e$ directly yields the balanced solution. 
This formulation transforms SO into a binary search problem while ensuring both computational efficiency and robustness.
\end{itemize}

\begin{figure}[t]
    \centering
    \includegraphics[width=0.92\linewidth]{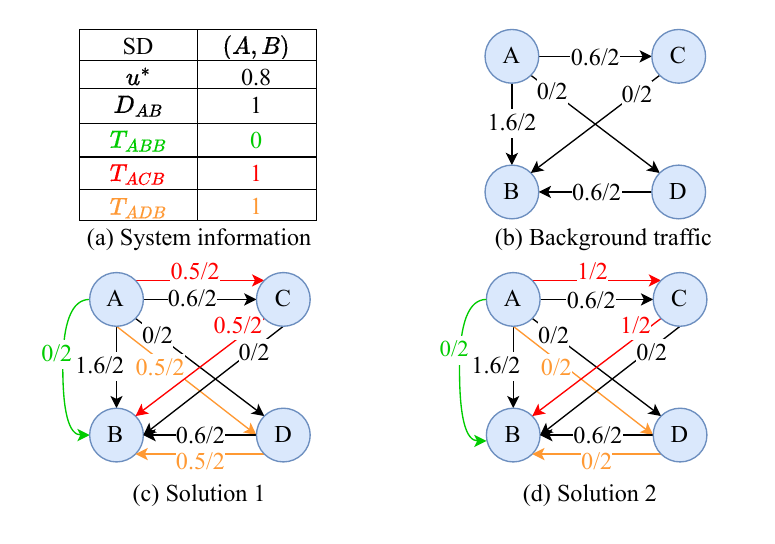}
    \caption{An illustration of the multi-solution phenomenon of SO. In this SO, using multiple split ratios will obtain the same MLU.}
 
    \label{fig:multi-solution}
\end{figure}

\noindent \textbf{Balanced binary search method for SO.} To efficiently solve the SO problem, we propose a balanced binary search algorithm (BBSM), as detailed in Algorithm \ref{ag:BBSM}. The algorithm is designed to leverage the characteristics of the problem for improved computational efficiency. Specifically, apart from the initialization step, all operations within BBSM have a time complexity of \( O(|V|) \). The binary search process is controlled by a threshold \( \epsilon \), typically set to \( 10^{-6} \), ensuring convergence within approximately \( \log_2(1 / \epsilon) = 20 \) iterations.

For the initialization phase, if the method in Equation (\ref{eq:Qij1}) is applied to calculate \( Q \), the time complexity reaches \( O(|V|^3) \). However, in practice, this complexity can be reduced to \( O(|V|) \) by maintaining a utilization matrix and updating the corresponding path utilization dynamically based on the selected SD. This practical implementation significantly reduces computational time overhead.

In contrast to the linear programming approach, whose computational complexity scales as $O(n^{3}) \approx O(|V|^{9})$ in a fully connected network, and which does not explicitly prioritize among multiple equally optimal solutions, the proposed BBSM demonstrates superior performance. With its lower computational complexity and its ability to identify well-balanced solutions among multiple feasible options, BBSM provides a more efficient and robust approach to the SO problem, particularly in large-scale networks.

\begin{algorithm}[t]
    \SetAlgoLined
    \KwIn{ $c$, $f_{skd}$, $s$, $d$, $D$.}
    \KwOut{Updated split ratio $f_{skd}$.}
    \BlankLine
    Initialize $Q$,  $u^{ub}$, $\underline u \gets 0$,
    $\bar{u} \gets u^{ub}$, continue $\gets$ \textbf{TRUE}; \\
    
    \While{continue}{
        $u \gets \frac{\bar{u} + \underline u}{2}$; \\
        Calculate $\bar{f}_{skd}^b(u)$\\
        \If{$\sum_{k \in V} \bar{f}_{skd}^b(u) \geq 1$}{
            $ \bar u \gets u$; \\
        }
        \Else{
            $\underline{u} \gets u$; \\
        }
        \If{$|\bar{u} - \underline {u}| < \epsilon$}{
            continue $\gets$ \textbf{FALSE}; \\
        }
    }
    Set $u \gets \bar{u}$; \\
    Set $f_{skd} \gets \bar{f}_{skd}^b(\bar u)$; \\
    \Return{$f_{skd}$};
    
    \caption{Balanced Binary Search Method (BBSM)}
    \label{ag:BBSM}
\end{algorithm}

\subsection{Detail of SSDO}\label{SD selection}

The \emph{SD Selection} component plays a critical role in determining the sequence of SDs for the \emph{Split Ratio Modification} component. A naive approach is to traverse all SDs in a fixed order. However, this is inefficient because many SDs have no impact on MLU, meaning their split ratios can be adjusted without affecting the optimization goal. As a result, computational resources are wasted on updating them.

To address this inefficiency, we leverage the mathematical relationship between MLU and SDs. As shown in Equation (\ref{eq:u}), the utilization rate of a link \(i\to j\) is influenced by up to \( 2|V| - 3 \) SDs. This implies that focusing on the SDs associated with the edges exhibiting the highest MLU can effectively reduce the MLU without needing to process all SDs. If any specific SD is restricted from using this link, it can simply be excluded from the calculation. 
\begin{equation}
    u_{ij}=\frac{\sum_{k\in V} f_{ijk} \cdot D_{ik} + \sum_{k\in V} f_{kij} \cdot D_{kj}}{c_{ij}}. 
    \label{eq:u}
\end{equation}

Based on this insight, the collaborative workflow of the \emph{SD Selection} and \emph{Split Ratio Modification} components is designed to prioritize efficiency.
\begin{enumerate}
    \item \textbf{\emph{SD Selection} component.} The \emph{SD Selection} component identifies the edges demonstrating the highest utilization. Subsequently, it calculates the SDs associated with these edges and organizes them into a processing queue using a specified prioritization rule (e.g., frequency of occurrence).
    \item \textbf{\emph{Split Ratio Modification} component.} The \emph{Split Ratio Modification} component processes the SDs in the queue one by one, adjusting their split ratios using BBSM to reduce MLU.
    \item \textbf{Termination check.} After processing all SDs in queue, SSDO evaluates whether the MLU has decreased. If the amount of MLU reduction is less than $\epsilon_0$, the algorithm terminates. Otherwise, the \emph{SD Selection} component recalculates the SD queue.
\end{enumerate}

The detailed steps of SSDO are summarized in Algorithm \ref{ag:SSDO}, which illustrates the interaction between two components. This collaborative design ensures that computational resources are focused on the most critical SDs, thereby improving the overall efficiency of the algorithm.

\begin{algorithm}[t]
    \SetAlgoLined
    \KwIn{$c$, $D$.}
    \KwOut{Optimized split ratios.}
    \BlankLine
    Initialize split ratios and calculate the utilization; \\
    Set continue $\gets$ \textbf{TRUE}; \\
    
    \While{continue}{
        Obtain the sequence of SDs using \emph{SD Selection} component; \\
        \For{each SD in the obtained sequence}{
            Call the \emph{Split Ratio Modification} component to update the split ratio; \\
            Update utilization;\\
        }

        \If{$\text{opt} -\max_{i,j\in V}u_{ij}  \leq \epsilon_0$}{
            continue $\gets$ \textbf{FALSE}; \\
        }
        \Else{        Update opt:
        \(\text{opt} \gets \max_{i,j\in V}u_{ij};\)
        
       }
    }
    \Return{Optimized split ratios.}
    \caption{Sequential Source-Destination Optimization (SSDO)}
    \label{ag:SSDO}
\end{algorithm}

\subsection{SSDO Deployment Strategies}
\noindent \textbf{Initialization modes.}  \label{initializaiton}
SSDO supports two initialization modes: hot-start and cold-start. In hot-start mode, SSDO uses TE configurations generated by other algorithms as the initial split ratios. The MLU in SSDO does not increase during the optimization process, guaranteeing that the solution quality is at least as good as the initial configuration. In the cold-start mode, SSDO initializes configurations according to predefined rules. Among various methods tested, directing all demands along the shortest path is identified as the most effective strategy due to its flexibility for subsequent optimization. Unless otherwise stated, all experiments in this paper adopt this cold-start method. For real-world deployment, a hybrid approach can be adopted: both hot-start and cold-start SSDO can be executed in parallel, and the system selects the best solution when the time limit is reached.

\noindent \textbf{Early termination.}  
SSDO achieves rapid MLU improvements during the early stages of optimization, making early termination a practical strategy, particularly in time-sensitive scenarios. This is especially effective in hot-start mode, particularly when initialized with DL-based solutions, which quickly generate feasible configurations for SSDO to refine with minimal computation. 
For deployment, an adaptive early termination mechanism can be implemented based on a predefined time threshold. This ensures that SSDO balances computation time and optimization quality efficiently.

\noindent \textbf{Path-based formulation.}  
For multi-hop scenarios (i.e., paths with three or more hops), SSDO must be extended to the path-based formulation, as detailed in Appendix~\ref{SSDO_path_form}. This formulation introduces incidence matrices to map split ratios to SDs, paths, and edges, enabling the model to capture multi-hop routing behaviors accurately. When the topology involves only one- or two-hop connections, adopting the path-based formulation is optional and can be decided based on the number of available candidate paths: if only a few exist, the path-based form can substantially reduce the problem size; otherwise, the original SSDO formulation remains preferable for its superior computational efficiency.

\section{Evaluation}\label{numerical-testing}
In this section, we present a comprehensive evaluation of SSDO. First, we outline the methodology and test system used in our experiments in \S\ref{Evaluation_1}. Next, we compare SSDO against other TE approaches, focusing on both TE quality and computational efficiency in \S\ref{Evaluation_TE}. Following this, \S\ref{Evaluation_failures} and \S\ref{demand change} evaluate SSDO’s effectiveness in managing link failures and adapting to dynamic traffic changes, respectively. Additionally, we assess the performance of SSDO on WAN in \S\ref{WAN}. The experiment about hot-start mode and early termination are detailed in \S\ref{solving time}. Finally, in \S\ref{Ablation}, we analyze the necessity of SSDO’s individual components through ablation studies.

\subsection{Methodology}\label{Evaluation_1}

\noindent\textbf{Topologies.} Our evaluation covers two types of topologies: Meta’s DCN \cite{roy_inside_2015}, including Top-of-Rack (ToR) and Point of Delivery (PoD) levels, and two WAN topologies, UsCarrier and Kdl, from the Internet Topology Zoo \cite{knightInternetTopologyZoo2011}. Shortest paths between SD pairs are precomputed using Yen’s algorithm \cite{abuzaid_contracting_2021}. Meta’s DCN topologies are modeled as complete graphs $K_n$ of sizes 4, 8, 155, and 367, corresponding to DB and WEB clusters under centralized TE. WAN topologies serve only as auxiliary demonstrations. Table \ref{tab:NETWORK TOPOLOGIES} summarizes the nodes, edges, and paths. For ToR-level DCNs, we test both per-pair 4-path limits and all-path settings.
\begin{table}[ht]
\centering
\begin{tabular}{ccccc}
\toprule
 & \#Type & \#Nodes & \#Edges & \#Paths \\

\midrule
\multirow{2}{*}{Meta DB}  & PoD-level DC & 4 & 12 & 3\\
&ToR-level DC & 155 & 23870 & 4 \\
&ToR-level DC & 155 & 23870 & 154 \\
\midrule
\multirow{2}{*}{Meta WEB}  & PoD-level DC & 8 & 56 & 7\\
 &ToR-level DC & 367 & 134322 & 4\\
 &ToR-level DC & 367 & 134322 & 366\\
 \midrule
  UsCarrier &WAN & 158 & 378 & 4\\
  \midrule
  Kdl & WAN & 754 & 1790 & 2\\

\bottomrule
\end{tabular}
\caption{Network topologies in our evaluation.}
\label{tab:NETWORK TOPOLOGIES}

\end{table}

\noindent\textbf{Traffic data.} In the study of Meta topologies, we utilize the publicly available one-day traffic trace provided by \cite{roy_inside_2015}. For the PoD-level topology, traffic traces are aggregated into 1-second snapshots of the inter-PoD traffic matrix, whereas for the ToR-level topology, aggregation is performed over 100-second intervals to generate the inter-ToR traffic matrix. For the UsCarrier and Kdl topologies from Topology Zoo, where no public traffic traces are available, we employ a gravity model \cite{applegate_making_nodate,roughan_experience_2003} to generate synthetic traffic.

\noindent\textbf{Baselines.}  
We select the following baselines to evaluate SSDO, with parameters chosen based on comprehensive considerations:  
(1) \textbf{LP-all:} Commercial LP solvers (Gurobi \cite{gurobi_optimization_llc_gurobi_2023}) directly solve TE, providing a theoretically optimal MLU.  
(2) \textbf{LP-top} \cite{namyar_minding_2022}: This method focuses on the top \(\alpha\%\) demands while routing the rest via shortest paths. Based on a trade-off between computational efficiency and solution quality, we select \(\alpha=20\) for all subsequent tests.  
(3) \textbf{POP} \cite{narayanan_solving_2021}: This method decomposes the optimization problem into $k$ subproblems, with each subproblem handling $1/k$ of the total demands while the capacity of each link is scaled down to $1/k$ of its original value. After balancing computational cost and performance, we set \(k=5\) for the evaluations.  
(4) \textbf{DOTE-m (DOTE \cite{perry_dote_nodate}, Figret \cite{liu_figret_2023}):} These methods take the traffic matrix as input and directly output the split ratios using a fully connected neural network. The models are trained with MLU as the loss function, optimizing traffic allocation to minimize congestion. In our experiments, we modify DOTE to take the current traffic matrix as input, referring to it as DOTE-m.
(5) \textbf{Teal} \cite{xu_teal_2023}: A reinforcement learning-based method using a shared policy network to allocate each SD's demands independently. The shared network significantly reduces the problem scale, making it suitable for large-scale networks.
The effectiveness of SSDO’s individual modules is assessed separately via ablation baselines in Section~\ref{Ablation}.

\noindent\textbf{Infrastructure and software.} Computational experiments are conducted on an Intel® Xeon® Platinum 8260 CPU with 1 TB of memory. Additionally, three NVIDIA GeForce RTX 4090 GPUs (each with 24 GB VRAM) are used for DL-based methods, including DOTE-m and Teal. These methods are implemented and evaluated using PyTorch 2.10, which is compatible with CUDA 12.1 \cite{paszke_pytorch_2019}. LP-based methods are evaluated using Gurobi 9.5.1 \cite{gurobi_optimization_llc_gurobi_2023}, with all reported times referring to TotalTime (including both model construction and solving, which is only marginally larger than RunTime in our setting). All implementations, including SSDO, are developed in Python 3.8.

\begin{figure*}[thbp] 
    \centering
    \includegraphics[width=\textwidth]{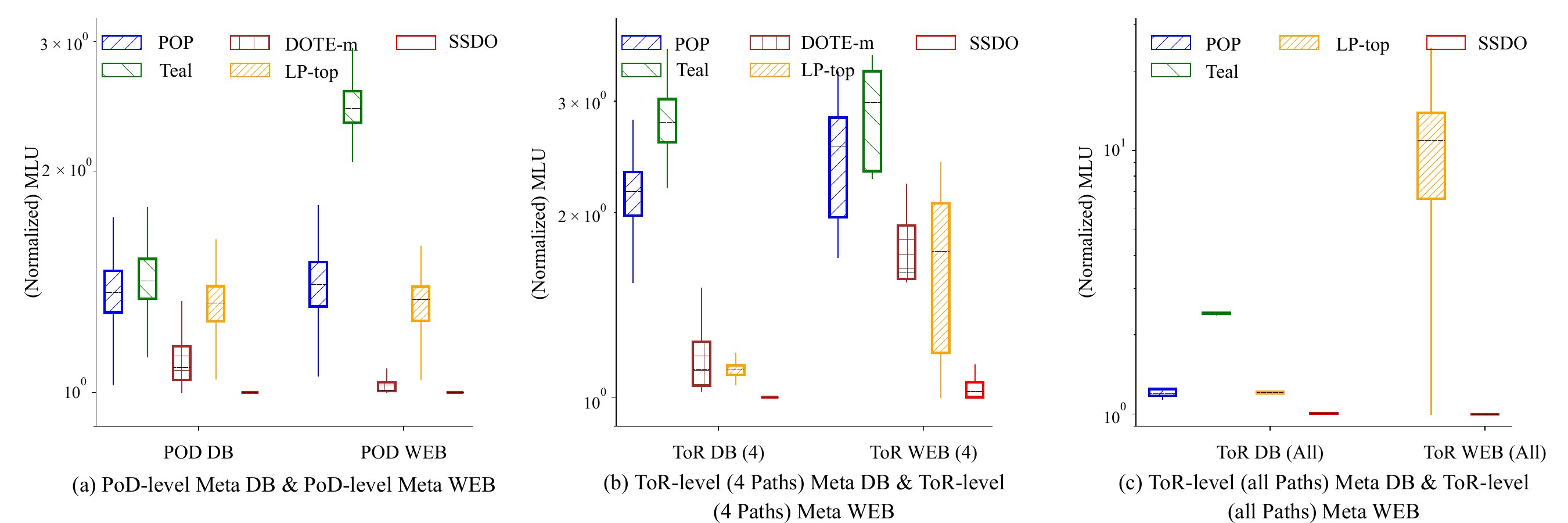} 
    \caption{TE quality performance of SSDO and other baseline. Methods order: POP, Teal, DOTE-m, LP-top, SSDO. In ToR-level (all paths) DB: DOTE-m failed. In ToR-level (all paths) WEB: DOTE-m, Teal and POP failed.} 
    \label{fig:TE Quality} 
\end{figure*}

\subsection{Comparison with Other TE Methods}\label{Evaluation_TE}

This section evaluates the TE performance and computation time across various topologies in Figure \ref{fig:TE Quality} and Figure \ref{fig:TE solver time}, focusing on normalized MLU relative to the LP-all method and computational time for each scheme. Both figures are presented on logarithmic scales for clarity. Notably, in the ToR-level Meta WEB topology (all paths), where LP-all fails to yield a feasible solution within the set time limitation (45,000 seconds), SSDO’s MLU serves as the normalization baseline. The results demonstrate SSDO’s exceptional balance between solution quality and efficiency, particularly in large-scale topologies. Key findings include:

 \noindent\textbf{LP-all:} Designed to provide optimal MLU solutions, LP-all serves as a benchmark for TE quality. However, its computation time increases exponentially with problem scale, becoming impractical even in medium-sized topologies. For instance, LP-all requires nearly 200s for the ToR-level Meta WEB (4 paths) topology and nearly 1,000s for the ToR-level Meta DB (all paths). In the ToR-level Meta WEB topology (all paths), LP-all fails to yield a feasible solution within time limitation, and thus its results are omitted from that topology.

\noindent\textbf{POP:} POP demonstrates unsatisfying TE performance due to its decomposition strategy, which isolates subproblems without accounting for coupling. While this approach can be effective for maximizing network flow, it is unsuitable for minimizing MLU. In the ToR-level WEB (4 paths) topology, POP’s MLU is 2.44$\times$ higher than SSDO’s. Furthermore, in the ToR-level Meta WEB topology (all paths), its solving time exceeds time limitations, making it infeasible for large-scale networks. Consequently, POP’s results are not included in Figure \ref{fig:TE Quality} or Figure \ref{fig:TE solver time} for this topology.

\noindent\textbf{LP-top:} LP-top improves upon LP-all by prioritizing the top $\alpha\%$ of demands, enabling better routing decisions for high-priority traffic. However, its simplistic handling of low-priority demands leads to unsatisfying configurations, especially in complex topologies like the ToR-level WEB (all paths), where its MLU is $10.93\times$ higher than SSDO’s. Additionally, LP-top’s computation time escalates with topology size, becoming impractical in large-scale scenarios.

\noindent\textbf{Teal:} While Teal achieves competitive efficiency in part of topologies, its TE quality remains unsatisfactory due to its design. Its shared policy structure struggles to capture the intricate demand couplings characteristic of DCNs. Moreover, Teal fails to provide feasible solutions in large-scale settings, where Video Random Access Memory (VRAM) limitations render it infeasible.

\noindent\textbf{DOTE-m:} DOTE-m quickly generates configurations in medium-scale topologies like ToR-level (4 paths). While its performance is inferior to SSDO, its fast inference speed provides an advantage. However, in large topologies, its fully connected network struggles with increased output dimensions and high VRAM consumption, limiting its scalability.

\noindent\textbf{SSDO:} SSDO achieves high-quality configurations across all tested topologies with competitive efficiency. For PoD level, despite its Python implementation, it reduces error rates below 1\% within 0.3s. For ToR-level WEB (4 paths), SSDO outperforms alternatives by reducing errors by 57\% in around 2s. In the challenging all-path Meta WEB topology, where most methods fail, SSDO completes optimization in 165s with robust accuracy. All tests use cold-start mode (\S\ref{initializaiton}). Notably, SSDO supports early termination, enabling high-quality solutions under time constraints (\S\ref{solving time}). Further improvements in implementation could enhance its performance.



\begin{figure*}[thbp] 
    \centering
    \includegraphics[width=\textwidth]{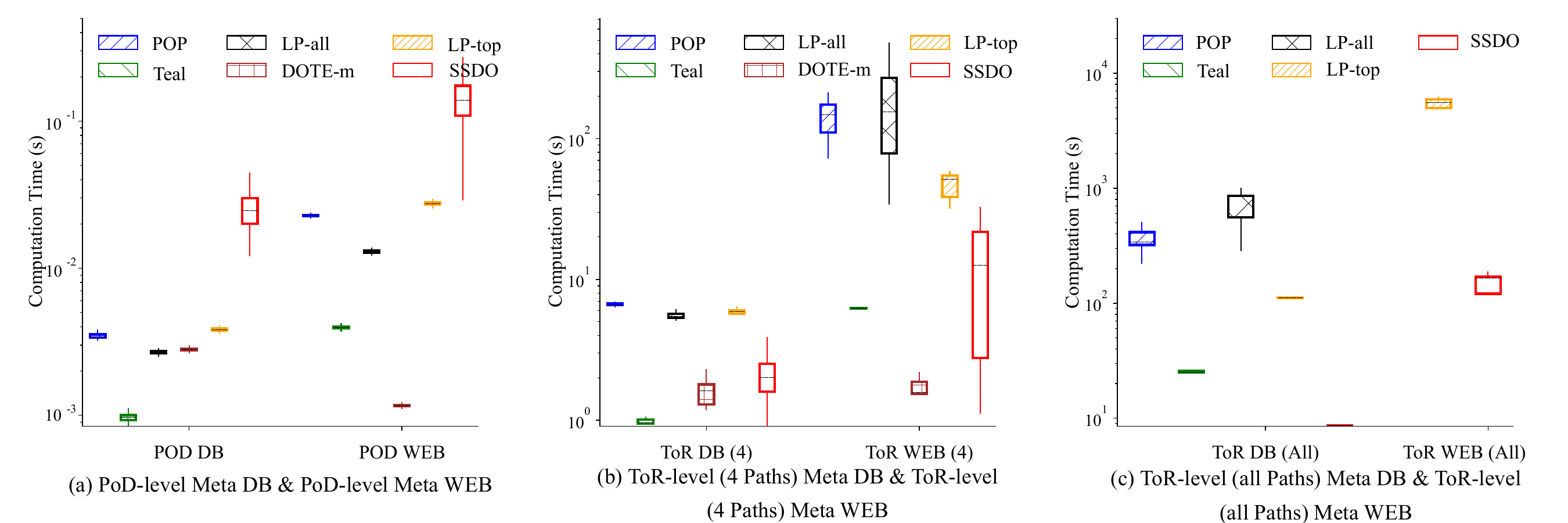} 
    \caption{Computation time performance. Methods order: POP, Teal, LP-all, DOTE-m, LP-top, SSDO. In ToR-level (all paths) DB: DOTE-m failed. In ToR-level (all paths) WEB: DOTE-m, Teal,LP-all and POP failed.} 
    \label{fig:TE solver time} 
\end{figure*}

\subsection{Coping with Network Failures}\label{Evaluation_failures}

Figure \ref{fig:failure} compares the performance of SSDO and other TE methods under different levels of random link failures in the ToR-level WEB topology (4 paths). Naturally, LP-all remains theoretically optimal even with a small number of failures, while maintaining stable MLU. Other LP-based methods exhibit poor performance, failing to meet practical requirements.

In addition, DOTE-m experiences a noticeable increase in MLU as failures grow. This degradation arises because its training data is derived from failure-free networks, rather than being intrinsic to DL-based approaches. When link failures occur, the mapping between traffic matrices and TE configurations shifts, leading to degraded performance. Indeed, recent work (e.g., Harp \cite{alqiamTransferableNeuralWAN2024}) has shown that carefully designed DL-based models can generalize to unseen topologies and even handle failures beyond the training distribution. More broadly, besides their inference speed, a key motivation for DL-based schemes is their potential to cope with prediction errors and uncertainty in traffic demands.

By contrast, SSDO achieves performance close to LP-all while maintaining strong adaptability and resilience to link failures. Unlike DL-based methods, SSDO does not require training data, making it a robust and practical choice for handling failures in dynamic network environments.

\begin{figure}[t] 
    \centering
    \includegraphics[width=0.9\linewidth]{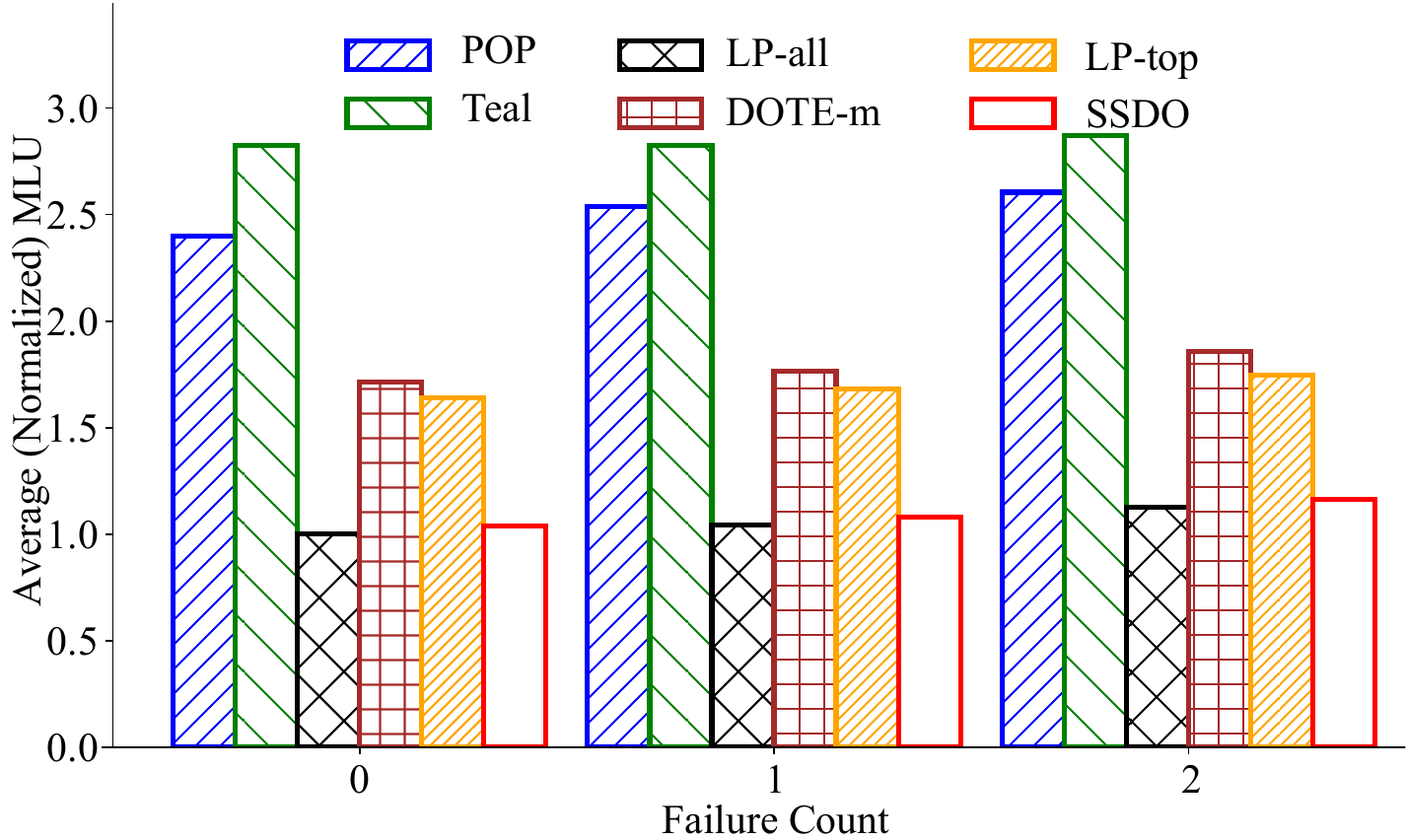} 
    \caption{Coping with different numbers of random link failures on ToR-level WEB (4 paths). The y-axis represents normalized MLU  using origin topology.}
    \label{fig:failure}
\end{figure}

\subsection{Robustness to Demand Changes}\label{demand change}

To assess the impact of temporal fluctuations on TE methods, we introduce different levels of variation into the traffic matrix. For each demand, we calculate the variance of its changes across consecutive time slots and scale it by factors of 2, 5, and 20. Using these scaled variances, we define zero-mean normal distributions, from which random samples are drawn and added to each demand in every time interval.

As shown in Figure \ref{fig:Temporal fluctuation}, SSDO maintains stable and high-quality performance across all fluctuation levels, demonstrating its robustness to temporal variations. LP-top and POP exhibit relatively stable performance, indicating that their optimization strategies are less sensitive to fluctuations. However, POP shows irregular variations, which stem from its  algorithmic design. Interestingly, LP-top’s performance slightly improves as fluctuations increase, likely because larger variations amplify the proportion of high-demand traffic, enabling LP-top to allocate resources more efficiently.

In contrast, DOTE-m and Teal experience a clear decline in performance as fluctuation levels increase. This degradation is likely caused by the growing discrepancy between the perturbed traffic matrices and the historical ones used for training, limiting generalization to unseen traffic patterns.


\begin{figure}[t] 
    \centering
    \includegraphics[width=0.95\linewidth]{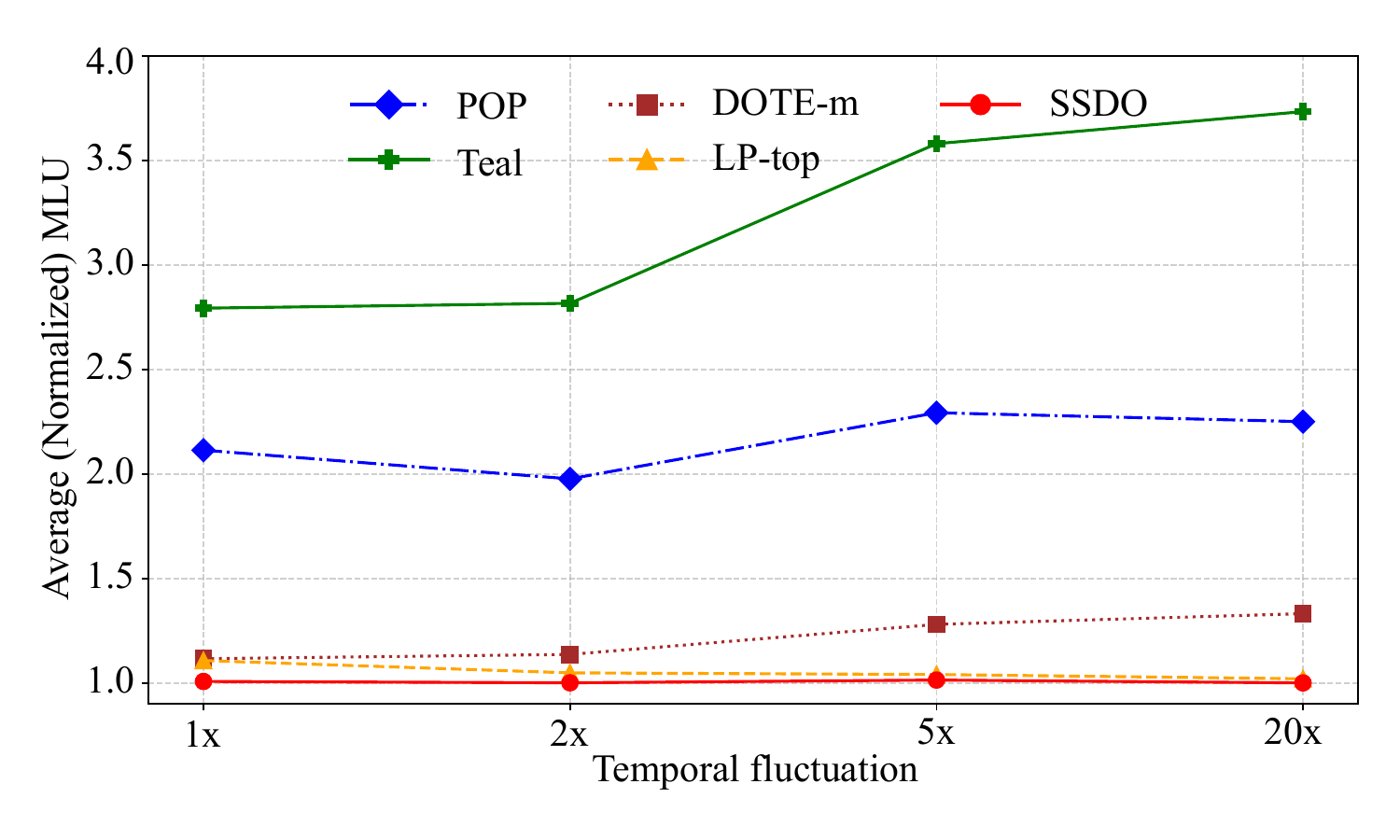} 
    \caption{Coping with Temporal fluctuation on Meta ToR-level DB (4 paths). The y-axis represents the MLU normalized by that of the LP-all using perturbed traffic matrix. }
    \label{fig:Temporal fluctuation}
\end{figure}

\subsection{SSDO for WAN}\label{WAN}
To demonstrate the generality of SSDO beyond DCNs, we further evaluate its path-based formulation (Appendix \ref{SSDO_path_form}) on two WAN topologies. Figure \ref{fig:SSDO for WAN} presents the results, comparing SSDO with various TE methods in terms of computation time and solution quality. SSDO consistently achieves high-quality solutions with competitive solving times, demonstrating its ability to generalize effectively to WAN settings.

In UsCarrier, SSDO achieves lower MLU than LP-based methods (POP, LP-top) while maintaining a solving time under one second, comparable to DL-based methods (DOTE-m, Teal). This efficiency highlights SSDO’s practicality in small-scale WANs. In KDL, SSDO reduces MLU by 9\% compared to DOTE-m and Teal while slightly outperforming POP. Although its solving time is marginally longer than DOTE-m, it remains significantly faster than LP-based methods. Notably, Teal’s solving time is higher than reported in prior work \cite{xu_teal_2023}, likely due to cases where it outputs all-zero split ratios, requiring additional corrections. 

\begin{figure}[bt]
    \centering
    \includegraphics[width=1\linewidth]{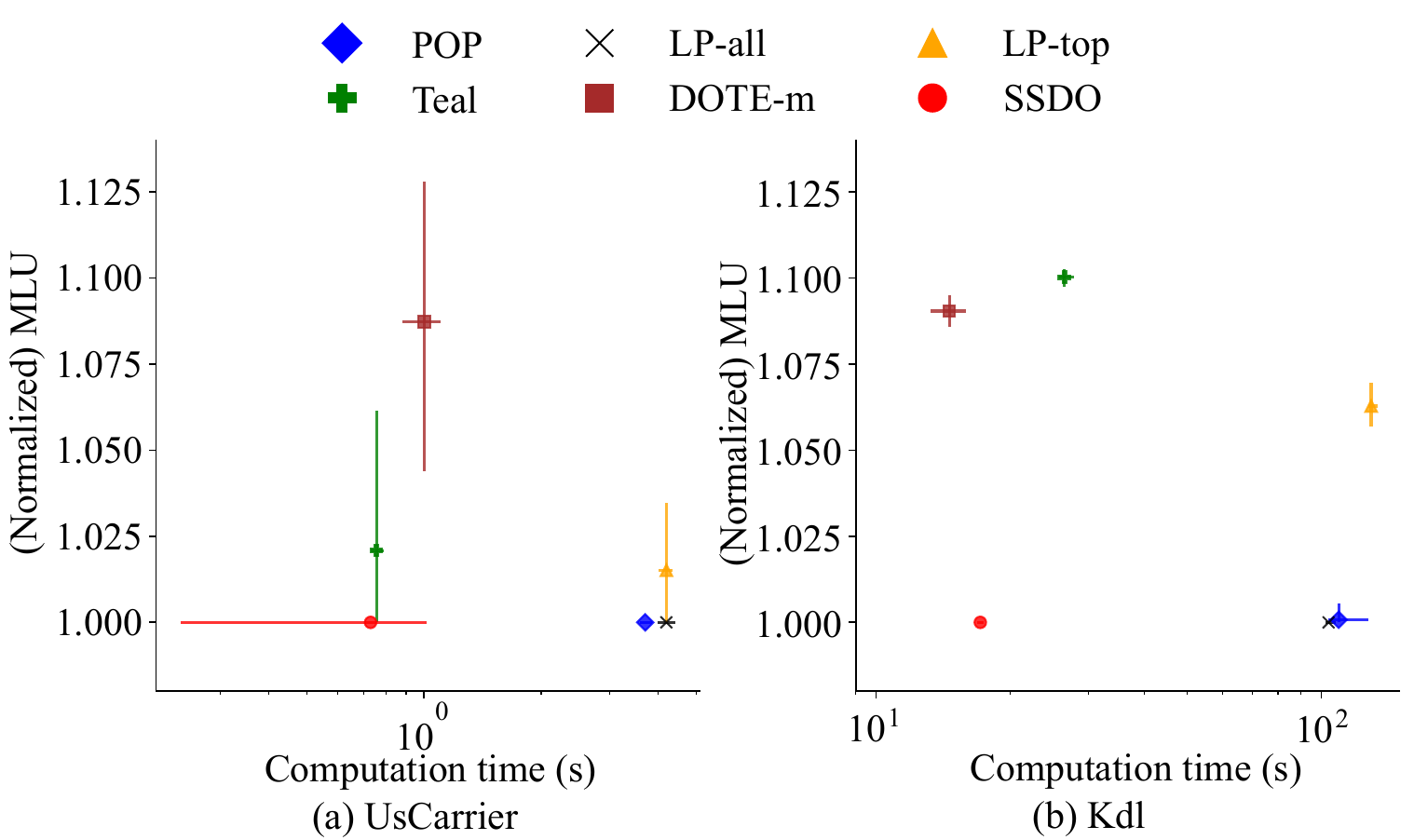}
    \caption{Performance of SSDO and baselines in WAN. The y-axis represents the normalized MLU. The x-axis represents the computation time (in seconds) on a logarithmic scale.}
    \label{fig:SSDO for WAN}
\end{figure}

\subsection{Hot-start Initialization and Early Termination in SSDO}\label{solving time}
Figure~\ref{fig:different time} shows the evolution of the MLU error relative to the optimal MLU throughout the SSDO optimization process. The y-axis represents the normalized error reduction, and the x-axis represents the normalized optimization time, ranging from 0 (start) to 1 (completion).
\begin{figure}[tbh]
    \centering
    \includegraphics[width=0.9\linewidth]{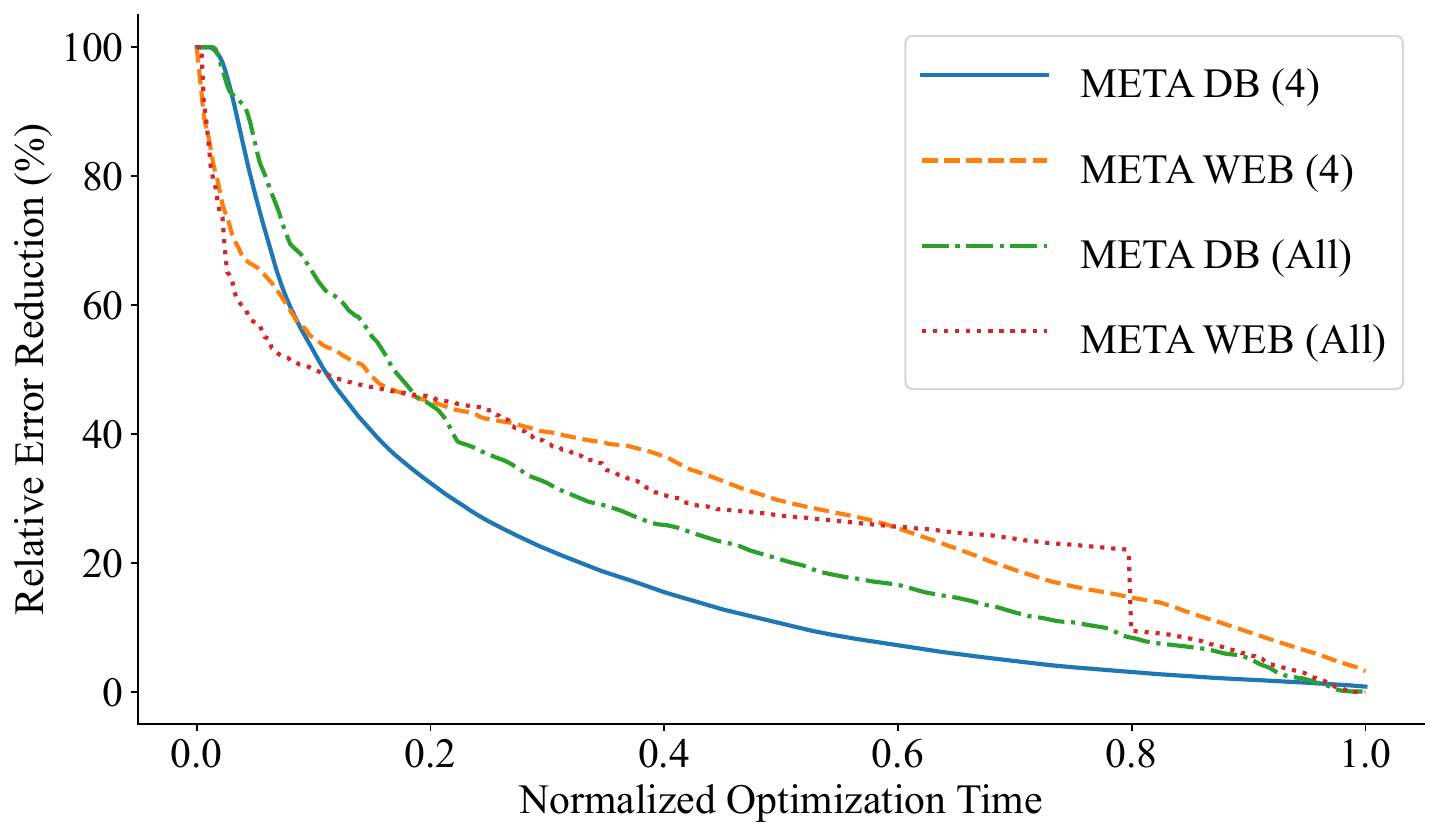}
    \caption{Relative error reduction of MLU in topologies. }
    \label{fig:different time}
\end{figure}
The results demonstrate that SSDO achieves rapid error reductions during the initial stages of optimization across all topologies. This characteristic provides strong support for the practicality of hot-start mode and early termination strategies, enabling high-quality solutions to be obtained with constrained computation time.

The effectiveness of hot-start SSDO is further validated in Appendix~\ref{appendix:solving_time_analysis}, which compares hot-start and cold-start modes. As shown in Figures~\ref{fig:MLU hot}-\ref{fig:time hot}, hot-start SSDO—initialized with DOTE-m solutions—outperforms DOTE-m and approaches the performance of cold-start SSDO, while requiring relatively less computation time.  	However, in some cases, cold-start SSDO completes optimization faster than hot-start mode due to the overhead of generating the initial solution by DOTE-m.  This suggests that in practical deployment, running both hot-start and cold-start SSDO in parallel and selecting the better-performing solution can further enhance efficiency. Additionally, Table~\ref{tab:SSDO_hot_tor_web} demonstrates that even with early termination, hot-start SSDO—leveraging DOTE-m’s solutions—reduces MLU by up to 35.9\% within just 3 seconds, confirming SSDO’s adaptability to different time constraints.

\subsection{Ablation Study of SSDO}\label{Ablation}
We perform an ablation study to assess the impact of SSDO’s key features on its overall performance.

\noindent\textbf{Design of BBSM.}
The BBSM accelerates the SO process and identifies globally beneficial solutions. In the SSDO/LP variant, subproblems are solved using LP solver (Gurobi), but split ratios are refined by BBSM to maintain consistency. Table~\ref{tab:time} shows that SSDO/LP is significantly slower than SSDO, demonstrating the efficiency of BBSM. Meanwhile, SSDO/LP-m employs split ratios calculated by Gurobi directly. As shown in Table~\ref{tab:mlu}, these ratios lead to a higher MLU, emphasizing the necessity of using balanced solutions.

\noindent\textbf{Design of SD Selection.}
SSDO optimizes SDs associated with edges of the highest real-time utilization, focusing on bottlenecks in each iteration. By contrast, SSDO/Static traverses all SDs per iteration. Table~\ref{tab:time} shows that SSDO/Static variant incurs substantially longer computation times, proving the efficiency of our prioritization strategy.

\begin{table}[h!]
\centering

\begin{tabular}{@{}lccc@{}}
\toprule
\textbf{Topology}     & \textbf{SSDO} & \textbf{SSDO/LP} & \textbf{SSDO/Static} \\ \midrule
PoD-level DB          & 0.03          & 0.15               & 1.27                 \\
PoD-level WEB         & 0.14          & 1.57               & 3.81                 \\
ToR-level DB (4)      & 2.16          & 202.28             & 184.37               \\
ToR-level WEB (4)     & 17.95         & 2796.84            & 3374.04              \\ \bottomrule
\end{tabular}
\caption{Comparison of computation Time (seconds) Across Variants}
\label{tab:time}
\end{table}

\begin{table}[h!]
\centering

\begin{tabular}{@{}lccc@{}}
\toprule
\textbf{Topology}     & \textbf{SSDO} & \textbf{SSDO/LP-m} \\ \midrule
PoD-level DB          & 1.00          & 1.10               \\
PoD-level WEB         & 1.00          & 1.44               \\
ToR-level DB (4)      & 1.01          & 3.41               \\
ToR-level WEB (4)     & 1.00          & 5.06               \\ \bottomrule
\end{tabular}
\caption{Comparison of MLU Across Variants}
\label{tab:mlu}
\end{table}


\section{Related Work}

\noindent\textbf{TE in DCNs and WANs.}
TE is critical for optimizing network performance, ensuring fairness, and preventing link overutilization in both DCNs and WANs. Hardware-based TE methods such as ECMP \cite{hoppsAnalysisEqualCostMultiPath2000,7017162} and WCMP \cite{zhouWCMPWeightedCost2014,cheng2020namp} are commonly employed to efficiently utilize bandwidth. However, these methods struggle with asymmetry and heterogeneity in traffic patterns. To overcome these challenges, SDN-based centralized TE systems \cite{wang_cope_2006,applegate_making_nodate} have gained popularity by addressing global optimization objectives such as MLU. While effective, scaling these systems to large, dynamic networks remains a significant challenge.

\noindent \textbf{Machine Learning in TE.}
Machine learning (ML)  has been applied in TE primarily for two purposes: prediction of traffic demand and direct configuration of TE. The first category uses predictive models to estimate future traffic based on historical data \cite{9213008,zhang2005optimal,kumar_semi-oblivious_nodate,Prophet}, which are then input into optimization algorithms to compute TE configurations. The second category learns a mapping from traffic to TE configurations, as demonstrated by methods like DOTE \cite{perry_dote_nodate} and others \cite{xu_teal_2023,valadarsky_learning_2017,liu_figret_2023}. Although these approaches leverage the ability of ML to model complex relationships, they face scalability challenges in large networks and struggle to handle unexpected traffic bursts, limiting their applicability in dynamic and large-scale networks.

\noindent \textbf{TE Acceleration.}
TE acceleration have been extensively studied to address the computational challenges of large-scale TE. For SDN environments, methods such as Teal \cite{xu_teal_2023} and POP \cite{narayanan_solving_2021} support both maximum flow and MLU minimization objectives, while NCFlow \cite{abuzaid_contracting_2021} is specifically tailored for maximum flow optimization. In hybrid SDN scenarios \cite{9213008,khorsandroo2021hybrid,silva2021hybrid}, Agarwal et al. \cite{6567024} proposed a greedy SDN switch placement approach combined with a fully polynomial-time approximation scheme to optimize traffic split ratios. Building on this, Guo et al. \cite{guo2021traffic,guo2014traffic,guo2021routing} introduced heuristic algorithms that jointly optimize OSPF link weights and SDN traffic splits, effectively reducing MLU in hybrid networks. Despite these advancements, achieving both efficiency and high TE quality remains a significant challenge in large-scale and dynamic network environments.

\section{Discussion}

\textbf{Analysis of optimality.} As noted in \cref{numerical-testing}, there remains a small but notable gap between the MLU achieved by SSDO and the theoretical optimum. This gap arises because SSDO may terminate when it encounters a particular situation, referred to as deadlock. In Appendix~\ref{sec:deadlock}, we provide an illustrative example of deadlock, define the phenomenon in detail, and demonstrate how it affects SSDO performance. We also discuss why deadlock rarely affects the reliability of SSDO.

\noindent\textbf{Analysis of objective.} SSDO is effective for MLU, since its bottleneck-driven monotonicity enables clean decomposition and efficient solution. Other objectives, such as throughput, lack these properties, though prior work like PCF~\cite{jiangPCFProvablyResilient2020} shows they can be related to MLU within a unified framework. Thus, SSDO’s strong guarantees hold for MLU, while extensions to other metrics remain possible only with approximation.  

\noindent \textbf{Analysis of Implementation:} In software-defined TE systems \cite{xu_teal_2023}, as described in \cref{sec:system}, SSDO is used within the TE controller to solve optimization problems. The TE controller takes current traffic demands and network topology as inputs, and produces traffic allocations as output. Recently, some DL-based systems have begun using historical traffic data as input. We believe SSDO could potentially be applied to these systems.

\section{Conclusion}
In this work, we introduce SSDO, a novel TE acceleration algorithm designed for large-scale DCNs. SSDO employs a sequential subproblem-solving strategy, where each subproblem optimizes the split ratios for a specific source-destination (SD). The subproblem order is dynamically adjusted based on real-time utilization to accelerate convergence. Each subproblem is solved using the Balanced Binary Search Method (BBSM), which efficiently identifies the most balanced and MLU-minimizing solution. To further improve efficiency, SSDO supports hot-start initialization, leveraging existing TE solutions as starting points, and early termination, ensuring high-quality solutions within limited computation time. Experimental results demonstrate that SSDO significantly outperforms existing methods, achieving superior TE quality while maintaining competitive computation efficiency. These features make SSDO a scalable and robust solution for large-scale TE in real-world networks.


\newpage
\section{ACKNOWLEDGMENTS}
We thank our anonymous reviewers for their insightful comments and Ryan Beckett for shepherding this work. We also thank Zhangheng Liu, Haozhe Li, Fanfan Li, Yuqian Ying, for their helpful feedback. This work was supported by “the Fundamental Research Funds for the Central Universities”.
\bibliographystyle{plain}
\bibliography{reference}

\newpage

\section*{Appendix}
\appendix

\section{Traffic Engineering in Path Form}\label{Notation for multi}

The traffic engineering (TE) problem in path form represents the flow distribution across candidate paths for each source-destination (SD). This formulation reduces the number of decision variables in constrained scenarios but requires additional structures to map paths, edges, and nodes.

\subsection{Notations}
\begin{itemize}
    \item \( G = (V, E, c) \): The network topology, where:
        \begin{itemize}
            \item \( V \): The set of nodes.
            \item \( E \): The set of edges.
            \item \( c_e \): The capacity of link \( e \in E \).
        \end{itemize}
    \item \( D_{sd} \): The traffic demand from source \( s \) to destination \( d \), expressed as a scalar value.
    \item \( P_{sd} \): The set of candidate paths between source \( s \) and destination \( d \). Each path \( p \in P_{sd} \) consists of a sequence of links.
    \item \( f_p \): The split ratio for path \( p \in P_{sd} \), representing the fraction of \( D_{sd} \) allocated to path \( p \). It satisfies  $\sum_{p \in P_{sd}} f_p = 1$.

\end{itemize}

\subsection{Optimization  Model}\label{optimization_problem}

The goal of TE is to minimize the maximum link utilization (MLU), ensuring balanced traffic distribution and avoiding congestion. The problem is formulated as follows:
\begin{align}
    \min_{f_p} \quad & \max_{e \in E} \frac{\sum_{s,d \in V} \sum_{p \in P_{sd}, e \in p} D_{sd} \cdot f_p}{c_e}, \label{eq:mlu_objective} \\
    \text{s.t.} \quad & \sum_{p \in P_{sd}} f_p = 1, \quad \forall s, d \in V, \label{eq:flow_conservation} \\
    & 0 \leq f_p \leq 1, \quad \forall s, d \in V, \forall p \in P_{sd}. \label{eq:split_ratio_bounds}
\end{align}

Equation~\eqref{eq:mlu_objective} minimizes the MLU across all network links. Equation~\eqref{eq:flow_conservation} ensures that the total split ratios sum to one for each SD, while Equation~\eqref{eq:split_ratio_bounds} enforces non-negativity and normalization constraints on the split ratios.

\section{SSDO in Path Form}\label{SSDO_path_form}

The Sequential Source-Destination Optimization (SSDO) minimizes the MLU \( u \) by iteratively adjusting path split ratios \( f_p \). The process consists of the following steps:

\begin{enumerate}
    \item \textbf{Initialization}:
    \begin{itemize}
        \item Set initial split ratios \( f_p \) for all paths, ensuring:
        \[
        \sum_{p \in P_{sd}} f_p = 1, \quad \forall s, d \in V.
        \]
        \item Compute the initial link utilization by
     $$
           U[e] = \sum_{s, d \in V} \sum_{p \in P_{sd}, e \in p} \frac{D_{sd} f_p}{c_e}.
           \label{Ue}
    $$
    \item Set initial $u_{prev}$:  $$u_{\text{prev}}=\max_{e \in E} U[e]. $$
    \end{itemize}

    \item \textbf{Identify Congested Edges}:
    \begin{itemize}
        \item Identify the set of edges \( E_{\text{max}} \subseteq E \) with utilization equal to the maximum \( u \):
        \[
        E_{\text{max}} = \{e \in E \mid U[e] = u_{prev}\}.
        \]
    \end{itemize}

    \item \textbf{Map to SD}:
    \begin{itemize}
        \item For each edge \( e \in E_{\text{max}} \), identify the set of SD \( (s, d) \) whose paths \( P_{sd} \) traverse \( e \).
    \end{itemize}

    \item \textbf{Update Split Ratios Using PB-BBSM}:
    \begin{itemize}
        \item For each identified SD \( (s, d) \), apply the Path-Based Balanced Binary Search Method (PB-BBSM) to update the split ratios \( f_p \) for paths \( p \in P_{sd} \).
        \item The detailed steps are provided in Section~\ref{sec:PB-BBSM}.
    \end{itemize}

    \item \textbf{Recompute Link Utilization}:
    \begin{itemize}
        \item Recalculate \( U[e] \) for all \( e \in E \) using the updated \( f_p \).
        \item Update the maximum link utilization \( u \):
        \[
        u = \max_{e \in E} U[e].
        \]
    \end{itemize}

    \item \textbf{Convergence Check}:
    \begin{itemize}
        \item If the reduction in \( u \) satisfies:
        \[
        |u_{\text{prev}} - u| \leq \epsilon_0,
        \]
        terminate the algorithm and return the optimized split ratios \( f_p \) and the minimized \( u \).
        \item Otherwise, set \( u_{\text{prev}} = u \), return to Step 2, and continue the iterations.
    \end{itemize}
\end{enumerate}

\section{Path-Based Balanced Binary Search Method}
\label{sec:PB-BBSM}

PB-BBSM adjusts the split ratios \( f_p \) for a given SD \( (s, d) \) to minimize \( u \), while ensuring traffic conservation. The algorithm is shown in Algorithm \ref{ag:PBBBSM}.

\begin{algorithm}[h]
    \SetAlgoLined
    \KwIn{Utilization matrix $U$, source $s$, destination $d$, demand matrix $D$, candidate paths $P_{sd}$, tolerance $\epsilon$.}
    \KwOut{Optimal split ratios $f_p$ for paths $p \in P_{sd}$.}
    \BlankLine
    Initialize $\underline{u} \gets 0$, $\bar{u} \gets \max(U)$; \\
    Initialize split ratios for all paths in $P_{sd}$; \\
            \For{each path $p \in P_{sd}$}{
            \(
            R[e] = U[e] - \frac{D_{sd} f_p}{c_e}, \forall e \in p;
            \)\\}
    \While{$\bar{u} - \underline{u} > \epsilon$}{
        $u_{\text{mid}} \gets \frac{\underline{u} + \bar{u}}{2}$; \\
        \For{each path $p \in P_{sd}$}{
            \(
            \bar{f}_p = \min_{e \in p} \frac{(u_{\text{mid}} - R[e]) \cdot c_e}{D_{sd}};
            \)  
             $\bar{f}^b_p \gets \max(\bar{f}_p, 0)$; \\
        }
        \If{$\sum_{p \in P_{sd}} \bar{f}^b_p > 1$}{
            $\bar{u} \gets u_{\text{mid}}$; \\
        }
        \Else{
            $\underline{u} \gets u_{\text{mid}}$; \\
        }
    }
 \For{each path $p \in P_{ij}$}{
                \(\bar{f}_p = \min_{e \in p} \frac{(\bar{u} - R[e]) \cdot c_e}{D_{sd}};
            \)  
             $\bar{f}^b_p \gets \max(\bar{f}_p, 0)$ ;
 
 }    
    \(
    f_p \gets \frac{\bar{f}_p^b}{\sum_{p \in P_{sd}} \bar{f}_p^b}, \quad \forall p \in P_{sd};
    \)\\
    \Return{$f_p$};

    \caption{Path-Based Balanced Binary Search Method (PB-BBSM)}
    \label{ag:PBBBSM}
\end{algorithm}

\section{Monotonicity of Upper Bound of the Split Ratio}
\label{appendix:monotonicity}

This appendix establishes the nondecreasing property of \(\bar{f}_{skd}(u)\) with respect to the MLU parameter \(u\).

\noindent\textbf{Notation recap.}  
From Equations~\eqref{eq:Tikj} and~\eqref{eq:fikj} in the main text:
\begin{equation*}
    T_{skd}(u) = 
    \begin{cases} 
        \min\bigl\{u c_{sk} - Q_{sk},\ u c_{kd} - Q_{kd}\bigr\}, & k \in P_{sd},\ k \neq d, \\[6pt]
        u c_{sd} - Q_{sd}, & k = d,
    \end{cases}
\end{equation*}
and
\begin{equation*}
    \bar{f}_{skd}(u) = \frac{T_{skd}(u)}{D_{sd}},
\end{equation*}
where
\begin{itemize}
    \item \(u \in \mathbb{R}_{\geq 0}\) is the candidate MLU value,
    \item \(c_e \geq 0\) denotes the link capacity of edge $e$,
    \item \(Q_{ij} \geq 0\) represents the background traffic on link \((i,j)\),
    \item \(D_{sd} > 0\) is the demand from source \(s\) to destination \(d\).
\end{itemize}

\begin{theorem}
For all \(s,d,k\in V\), the function \(\bar f_{skd}(u)\) is nondecreasing over \(u \in [0,+\infty)\).
\end{theorem}

\begin{proof}
Consider a link \((i,j)\). The term \(g_{ij}(u) = u c_{ij} - Q_{ij}\) is affine in \(u\) with slope \(c_{ij}\ge 0\). Hence \(g_{ij}(u)\) is nondecreasing.  

If \(k=d\), then
$$
T_{skd}(u) = g_{sd}(u),
$$
which is nondecreasing.  
If \(k\neq d\), then
$$
T_{skd}(u) = \min\{g_{sk}(u), g_{kd}(u)\},
$$
the pointwise minimum of two nondecreasing functions, which is also nondecreasing.  

Since \(D_{sd}>0\), dividing by \(D_{sd}\) preserves monotonicity. Therefore,
\(\bar f_{skd}(u) = T_{skd}(u)/D_{sd}\) is nondecreasing in \(u\).
\end{proof}

\paragraph{Remark.}
A finite sum of nondecreasing functions remains nondecreasing, so 
\(\sum_{k\in V}\bar f_{skd}(u)\) is also nondecreasing. This property justifies the feasibility check and the binary search procedure used in the main text.

\section{Hot-start and Early Termination Analysis}
\label{appendix:solving_time_analysis}

This section evaluates the performance of hot-start SSDO and the effectiveness of early termination strategies. Experiments were conducted on the ToR-level WEB topology (4 paths) topology, comparing hot-start SSDO (SSDO-hot) with cold-start SSDO (SSDO-cold) and DOTE-m. Additionally, we analyze the effect of early termination in hot-start to highlight its practicality for time-sensitive network.

\begin{figure}[hbt]
    \centering
    \includegraphics[width=1\linewidth]{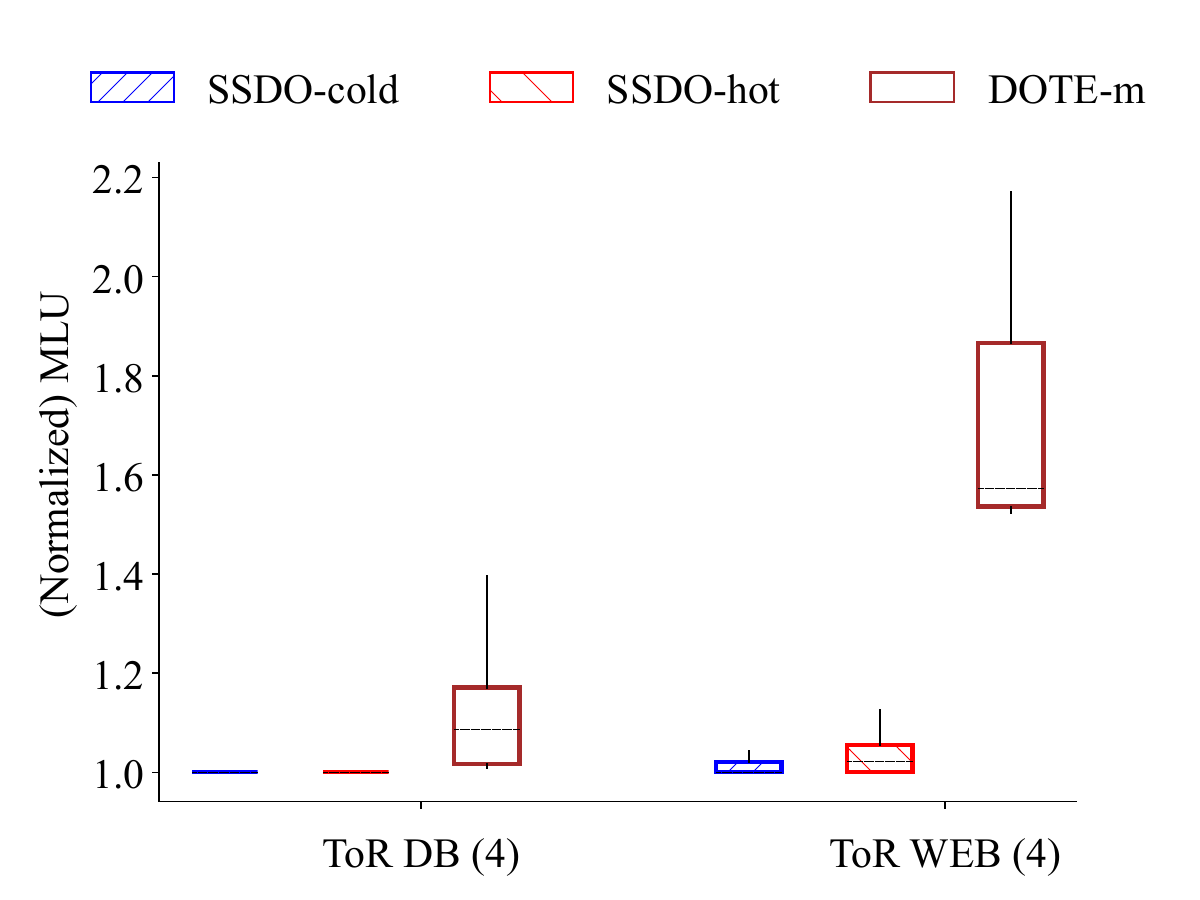}
    \caption{Comparison of SSDO-hot, SSDO-cold, and DOTE-m in MLU for ToR-level (4 paths) topologies.}
    \label{fig:MLU hot}
\end{figure}

\begin{figure}[hbt]
    \centering
    \includegraphics[width=1\linewidth]{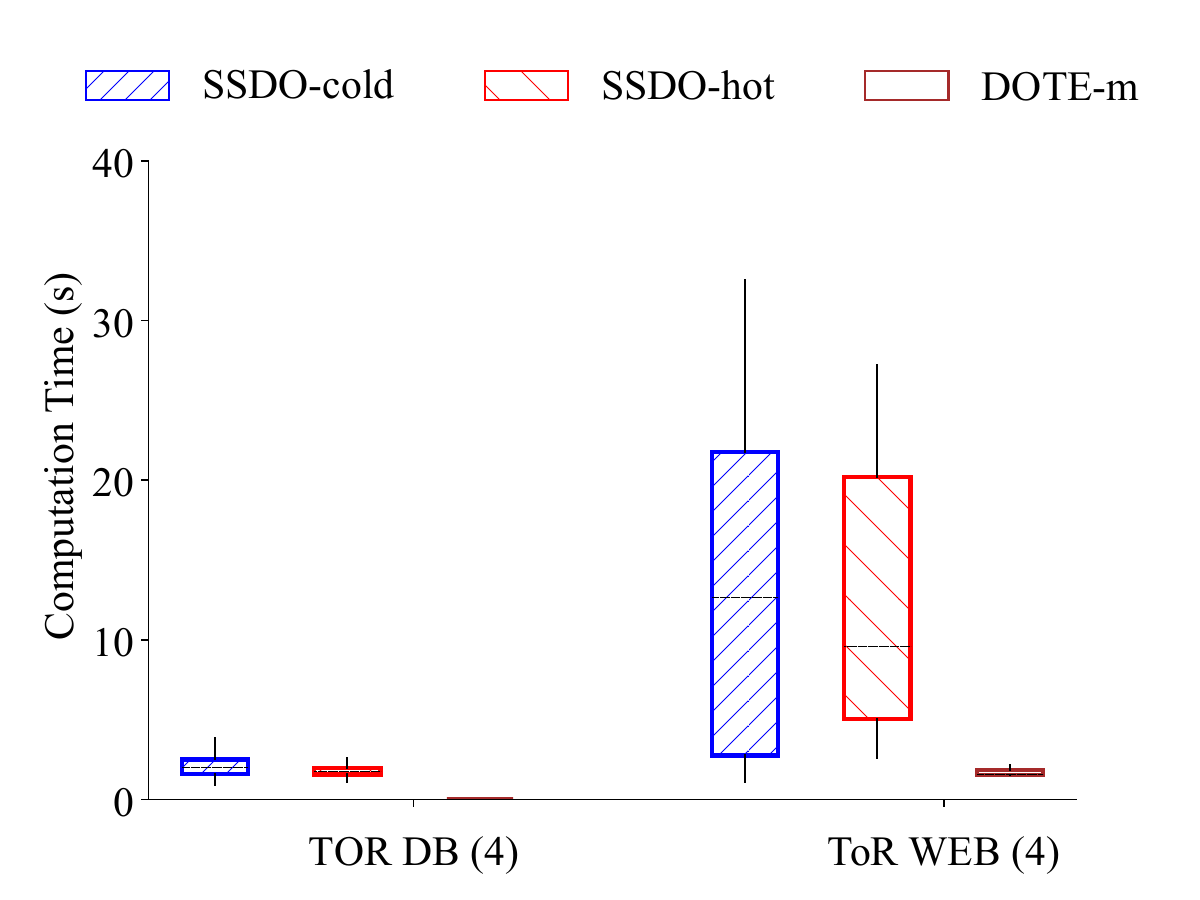}
    \caption{Comparison of SSDO-hot, SSDO-cold, and DOTE-m in computation time for ToR-level (4 paths) topologies.}
    \label{fig:time hot}
\end{figure}

\subsection{Effectiveness of Hot-start Mode}

In hot-start mode, SSDO initializes with solutions generated by DOTE-m, while in cold-start mode, the initial split ratios are determined based on the shortest-path strategy, as described in \S~\ref{initializaiton}. Figure~\ref{fig:MLU hot} compares the MLU achieved by SSDO-hot, SSDO-cold, and DOTE-m. The results show that SSDO-hot consistently outperforms DOTE-m and achieves performance close to SSDO-cold. Figure~\ref{fig:time hot} presents the computation time comparison. Although SSDO-hot includes the time required for DOTE-m to generate the initial solution, it runs faster than SSDO-cold in most cases. This highlights the advantage of hot-start mode in efficiently refining existing solutions while reducing computational cost.

\subsection{Effectiveness of Early Termination in Hot-start Mode}
To evaluate early termination, we track MLU reduction in SSDO-hot over time. Table~\ref{tab:SSDO_hot_tor_web} shows that SSDO-hot achieves substantial improvements within a few seconds. For instance, case 8 reduces MLU by 24.2\% in 5 seconds, while cases 1 and 2 reach optimal solutions even faster. These results demonstrate that early termination in hot-start scenarios effectively balances solution quality and runtime.
\begin{table}
    \centering

    \begin{tabular}{c|cccc}
        \hline
        Case & 0s & 3s & 5s & 10s \\
        \hline
        1 & 1.5637 & 1.0000 & 1.0000 & 1.0000 \\
        2 & 1.5225 & 1.0000 & 1.0000 & 1.0000 \\
        3 & 1.5384 & 1.1842 & 1.1412 & 1.0545 \\
        4 & 1.9564 & 1.4177 & 1.3047 & 1.1329 \\
        5 & 1.8368 & 1.6098 & 1.5286 & 1.4208 \\
        6 & 1.5824 & 1.2440 & 1.2035 & 1.0564 \\
        7 & 1.5291 & 1.2353 & 1.1643 & 1.0000 \\
        8 & 2.1710 & 1.7314 & 1.6415 & 1.4610 \\
        \hline
    \end{tabular}
        \caption{Normalized MLU reduction over time in SSDO-hot for ToR-level WEB (4 paths) topology.}
    \label{tab:SSDO_hot_tor_web}
\end{table} 

\subsection{Summary of Hot-Start and Early Termination Advantages}
	The results demonstrate SSDO’s robustness in handling strict computational constraints. Hot start accelerates optimization by leveraging existing solutions, while early termination ensures high-quality results within limited time. These strategies enable SSDO to efficiently adapt to real-time performance demands, making it a practical solution for dynamic and time-sensitive environments.

\section{Deadlock of SSDO}\label{sec:deadlock}

\subsection{Definition of Deadlock in SSDO}
The definition of Deadlock is as follows.

\begin{definition}[Deadlock in SSDO]
We call the current split‐ratio configuration of SSDO to be in a \emph{deadlock} if the following two conditions hold:
\begin{itemize} 
    \item For every Source-destination (SD) pair, any adjustment to its split ratios, while keeping the split ratios of all other SD demands fixed, fails to reduce the current MLU.
    \item There nevertheless exists some feasible traffic engineering configuration (e.g., by jointly adjusting split ratios for multiple SD demands) that achieves a strictly lower MLU than the current configuration.
\end{itemize}
\end{definition}

\subsection{Illustrative Example}
An anonymous reviewer provided a compelling example that highlights this phenomenon. 
Consider a directed ring with $n=8$ nodes labeled $A$ through $H$, with clockwise edges 
$\{AB, BC, CD, DE, EF, FG, GH, HA\}$, each of capacity $1$. 
In addition, we add ``skip'' edges that connect every second node 
(e.g., $AC, BD, CE, DF, EG, FH, GA, HB$), each of infinite capacity. 
Each demand is placed between adjacent nodes in the clockwise direction, and each demand has two possible paths: 1) the direct one-hop edge (e.g., $AB$), 2) a long detour around the rest of the ring (e.g., $ACDEFGHB$).
The demand for each clockwise pair of nodes is $1/5$, e.g., $A \to B, B \to C, \ldots, H \to A$. 
The topology, including paths and demand information, is illustrated in Figure~\ref{fig:ring-example}.

\begin{figure}[tbh]
\centering
\begin{tikzpicture}[scale=1.0, every node/.style={circle, draw, minimum size=0.7cm}]
  \node (A) at (0,2) {A};
  \node (B) at (2,2) {B};
  \node (C) at (4,2) {C};
  \node (D) at (4,0) {D};
  \node (E) at (4,-2) {E};
  \node (F) at (2,-2) {F};
  \node (G) at (0,-2) {G};
  \node (H) at (0,0) {H};

  \draw[->, thick] (A) -- (B);
  \draw[->, thick] (B) -- (C);
  \draw[->, thick] (C) -- (D);
  \draw[->, thick] (D) -- (E);
  \draw[->, thick] (E) -- (F);
  \draw[->, thick] (F) -- (G);
  \draw[->, thick] (G) -- (H);
  \draw[->, thick] (H) -- (A);

  \draw[->, dashed, thick] (A) to[bend left=30] (C);
  \draw[->, dashed, thick] (B) to[bend left=30] (D);
  \draw[->, dashed, thick] (C) to[bend left=30] (E);
  \draw[->, dashed, thick] (D) to[bend left=30] (F);
  \draw[->, dashed, thick] (E) to[bend left=30] (G);
  \draw[->, dashed, thick] (F) to[bend left=30] (H);
  \draw[->, dashed, thick] (G) to[bend left=30] (A);
  \draw[->, dashed, thick] (H) to[bend left=30] (B);

\end{tikzpicture}
\caption{Directed ring topology with $n=8$ nodes and skip edges. Each clockwise pair of nodes has a demand of $1/5$, with two paths available for each demand: the direct one-hop edge (e.g., $AB$) and a long detour around the ring (e.g., $ACDEFGHB$). The topology includes both solid edges (capacity $1$) and dashed skip edges (infinite capacity).}
\label{fig:ring-example}
\end{figure}

Under the deadlock configuration, each SD routes its entire demand along the long detour (bypassing its direct one-hop edge). 
Let $D=1/(n-3)$. 
Each detour traverses exactly $n\!-\!3$ unit-capacity ring edges, and by symmetry each ring edge lies on $n\!-\!3$ distinct detours; hence the load on every unit-capacity ring edge equals $(n-3) D=1$, so $\mathrm{MLU}=1$. 
Changing the split ratios of a \emph{single} SD only removes traffic from its $n\!-\!3$ detour edges and shifts it onto its direct edge, which then exceeds capacity. As a result, the maximum utilization remains at least~1. 
Only a coordinated adjustment of many SDs can simultaneously reduce the load on all ring edges. 
By contrast, the global optimum sends each SD on its direct one-hop edge, yielding a per-edge load of $D=1/(n-3)$ and thus $\mathrm{MLU}=1/(n-3)$.

This example demonstrates that SSDO may, in principle, terminate at a deadlock and miss the optimum. However, it’s important to emphasize two points:
\begin{itemize}
    \item \textbf{Pathological Initialization:} The scenario described relies on a pathological initialization where all traffic is routed along detour paths. This initialization is highly unlikely in practice. In \S\ref{initializaiton}, we introduce a cold-start initialization strategy that systematically avoids these pathological cases. 
\item \textbf{Proximity to Optimality:} In our numerical tests, SSDO almost always encounters deadlocks. However, even when deadlock occurs, the resulting solution is typically very close to the optimum. According to our cold-start initialization method, the initial solution will not be a deadlock (if all demands are routed through their one-hop path, but the MLU cannot be reduced by adjusting a single SD demand, the configuration is already optimal). Therefore, when deadlock occurs, it is typically far from the initial solution.
\end{itemize}

\section{TE Systems}
\label{sec:system}

\begin{figure}[bth]
\centering
\includegraphics[width=0.46\textwidth]{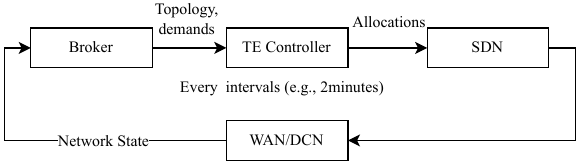}
\caption{Control loop of  traffic engineering. The TE controller periodically receives demand and topology inputs, solves the optimization problem, and updates router configurations through SDN.}
\label{fig:control_loop}
\end{figure}

Traffic engineering (TE) systems are integral to managing the flow of traffic. These systems allocate network demands to achieve high link utilization, and improve resilience to link failures.  In TE systems, traffic allocation is formulated as an optimization problem, where the goal is to achieve optimal network properties such as minimum MLU, maximum throughput. The core of these systems is typically a software-defined TE controller, as illustrated in Figure~\ref{fig:control_loop}.

The TE controller operates in a periodic control loop, as depicted in Figure~\ref{fig:control_loop}. It receives real-time traffic demands and network topology information from the bandwidth broker at regular intervals (e.g., every two minutes). The TE controller then solves the optimization problem based on these inputs and calculates the optimal traffic allocations. Finally, the traffic allocations are translated into router configurations that are deployed via Software-Defined Networking, updating the routing policies.

This periodic feedback loop ensures that the network continuously adapts to changing traffic patterns and evolving network conditions. The ability to dynamically adjust traffic distribution is essential for achieving efficient network utilization and maintaining quality of service.

\end{document}